\documentclass[runningheads]{llncs}

\usepackage{amsmath,amsfonts}
\usepackage{mathrsfs}
\usepackage{tikz-cd}% diagrams
\usepackage{hyperref}
\usepackage{comment}

\usepackage{graphicx}
\newcommand{\F}{{\mathbb F}}

\newcommand{\fd}{\mathscr{F}in\mathscr{D}ist}

\newcommand{\cA}{\mathcal{A}}
\newcommand{\cB}{\mathcal{B}}
\newcommand{\cC}{\mathcal{C}}
\newcommand{\cD}{\mathcal{D}}

\newcommand{\cH}{\mathcal{H}}

\newcommand{\sG}{\mathscr{G}}
\newcommand{\sF}{\mathscr{F}}
\newcommand{\sH}{\mathscr{H}}

\newcommand{\sP}{\mathscr{P}}

\newcommand{\sS}{\mathscr{S}}

\newcommand{\sC}{\mathscr{C}}

\newcommand\correspondingauthor{\thanks{Corresponding author.}}

\begin{document}

\title{Bridges connecting Encryption Schemes}
%
%\titlerunning{Abbreviated paper title}
% If the paper title is too long for the running head, you can set
% an abbreviated paper title here
%

\author{Mugurel Barcau\inst{1,2} \and Cristian Lupa\c scu
\inst{1,3} \and
Vicen\c tiu Pa\c sol \inst{1,2} \and George C. \c Turca\c s \inst{1,4} \correspondingauthor}
\authorrunning{M. Barcau et al.}
% First names are abbreviated in the running head.
% If there are more than two authors, 'et al.' is used.
%
\institute{certSIGN -- Research and Innovation, Bucharest, Romania 
\and
Institute of Mathematics ``Simion Stoilow" of the Romanian Academy \and
Ferdinand I Military Technical Academy, Bucharest, Romania \and Babe\c s-Bolyai University, Cluj-Napoca, Romania
\email{\{alexandru.barcau,cristian.lupascu,vicentiu.pasol,george.turcas\}@certsign.ro}}

\maketitle              % typeset the header of the contribution
\begin{abstract}
The present work investigates a type of morphisms between encryption schemes, called bridges. By associating an encryption scheme to every such bridge, we define and examine their security. Inspired by the bootstrapping procedure used by Gentry to produce fully homomorphic encryption schemes, we exhibit a general recipe for the construction of bridges. Our main theorem asserts that the security of a bridge reduces to the security of the first encryption scheme together with a technical additional assumption.
\end{abstract}

\section{Introduction}
The idea of switching ciphertexts encrypted using the same scheme from one secret key to another appears in the literature under the name of Proxy Re-Encryption (see \cite{Do03} and the references within). More recently, a general method of converting ciphertexts from one encryption scheme to another was introduced in \cite{DN21} under the name of Universal Proxy Re-Encryption. In practice, Proxy Re-Encryption between two distinct (arbitrary) schemes is very difficult to realize, as the general methods proposed in \cite{DN21} make use of hard to achieve protocols such as indistinguishability obfuscation. In this work, we focus on unidirectional such protocols between two distinct encryption schemes and call them {\it bridges} (see Definition \ref{bridge}).

Bridges are important tools in the context of Hybrid Homomorphic Encryption (see for example \cite{Rasta18} and \cite{Pasta21}), where the owner encrypts its data using a symmetric cipher and sends the encryption to a server together with his symmetric key encrypted under a homomorphic encryption scheme. The server first homomorphically performs the decryption circuit of the symmetric cipher to transform the initial ciphertext into one that allows homomorphic computation and then proceeds with the desired computations. The result of this computation can only be decrypted by the data owner. Apart from other possible applications, there is another motivation for studying these primitives which comes from the perspective of viewing encryption schemes in a categorical context, where bridges play the role of morphisms in an appropriate category.  

 In his remarkable work on fully homomorphic encryption, C. Gentry \cite{Ge09T} used a {\it Recrypt} procedure in order to transform a somewhat homomorphic encryption scheme into a fully homomorphic encryption scheme. To be precise, Gentry's \textit{Recrypt} algorithm takes as input a ciphertext together with certain encryptions of the secret key under a different key and evaluates homomorphically the decryption algorithm in order to produce an encryption of the same plaintext under the new key. Under the definition we propose, the \textit{Recrypt} algorithm is a bridge from a somewhat homomorphic encryption scheme to itself. The recipe can be extended to produce a bridge from any encryption scheme to any somewhat homomorphic encryption scheme that can correctly evaluate the decryption circuit of the former.

Perhaps connected to the same idea is the work in \cite{CPP16}, where maps between two encryption schemes are used to construct a $2$-party computation protocol, called an Encryption Switching Protocol (ESP). The examples proposed in \cite{CPP16} and \cite{CIL17} consist of two encryptions schemes over the same plaintext, which has a structure of a ring, and switching protocols between them. One of the schemes is homomorphic with respect to addition and the other is homomorphic for the multiplication. An ESP of this form can be used to construct a secure general $2$-party computational protocol. 

Switching between one encryption scheme to another, in order to securely perform a sequence of homomorphic operations, is a recurrent theme in the literature. In this respect, it is important to formally define and analyze the security implications of such protocols, which represents the main goal of the present work. We shall call a map (or a morphism) between encryption schemes satisfying certain properties a {\it bridge}. The terminology is borrowed from  \cite{BGGJ}, where the expression ``bridge between encryption schemes" is briefly used in reference to a hybrid solution for switching between FHE schemes in order to optimize performance of  certain homomorphic computations on encrypted data. 

\medskip 

\noindent\textbf{Our contribution} \,  In this paper, we first propose a general definition for a bridge, formalizing the conditions under which an algorithm that publicly transforms encrypted data from one scheme to another should perform. We provide a general recipe, inspired by Gentry's idea, for the construction of bridges and then apply it to give various examples. This general recipe can be modified in various ways and we demonstrate this by presenting a variant of it. We also present an additional example of a bridge that does not fall in the category of Gentry type bridges. We canonically associate to any bridge an encryption scheme and then define the security of a bridge as being the security of its associated  encryption scheme. This association is widely used in mathematics when someone needs to replace a morphism between two objects by an object. More precisely, it consists in substituting a map by its graph, whenever this is possible. We prove a general theorem (Theorem \ref{mt}) asserting that the security of a bridge reduces to the security of the first encryption scheme together with a technical additional assumption. We show that the latter technicality is in fact a natural condition by proving that bridges obtained using Gentry's \textit{Recrypt} idea satisfy this assumption (Proposition \ref{prop:proof}).
The security analysis provided here is finer than the corresponding security analysis made on Proxy Re-Encryption schemes. Our definition of IND-CPA security of the bridge and Cohen's HRA security definition (see \cite{Coh19}) are equivalent (all players are honest) and thus Cohen's simulatability theorem (see Theorem 5 of \cite{Coh19}) is vacuous in the case discussed in this paper. Our work is accompanied by three appendices. In the first two, we present examples of bridges (of different types). 
Comments on the performance of the implementations of these examples are to be found in the last appendix.

\medskip

\noindent\textbf{Organization} \, The article is organized as follows. Section 2 consists of some mathematical background and preliminaries about encryption schemes used in the rest of the article. It starts by recalling some terminology and theoretical facts about finite distributions. In the same section, we also give the definition of a bridge. The contributions in section 3 regard the security of a bridge between two encryption schemes. The main result of our paper (Theorem \ref{mt}) is proved in this section. In section 4, we show that Gentry's \textit{Recrypt} algorithm gives a general recipe for the construction of bridges. Using the main result from the previous section, we prove that bridges generated using this recipe are secure. The appendices are organised as follows. By representing the decryption circuit of a specific encryption scheme in four different ways, we give in appendix A, four different examples of bridges from the same encryption scheme to various FHE schemes. A bridge with empty bridge key,  not following the recipe presented in section 4, connecting the GM and SYY encryption schemes is exhibited in appendix B. Its security follows from results in section 3. The homomorphic evaluation of a comparison circuit is presented as an application to the latter bridge. In the last appendix of this article, we report on the results of several experiments involving the implementation of the bridges introduced in the first two appendices.

\section{Preliminaries}

In all our definitions, we denote the security parameter by $\lambda$. We say that a function $\mu: \mathbb{N} \rightarrow [0, + \infty)$
is a negligible function if for any positive integer $c$ there exists 
a positive integer $N_c$, such that $\mu (n) < \dfrac{1}{n^c}$ for all $n \geq N_c$.

\subsection{Finite Distributions}
\label{finitedist}

A finite probability distribution is a probability distribution with finite support. If $X$ is a finite distribution, we denote by $|X|$ its support. If $X$ and $Y$ are finite distributions, then a morphism $\varphi:Y \rightarrow X$ is a map of sets (still denoted by) $\varphi: |Y| \rightarrow |X|$ such that  

\begin{equation*}
{\rm Pr} \{X=x\} = \sum_{y \in \varphi^{-1}(x)} {\rm Pr} \{Y=y\}.
\end{equation*}

\noindent for all $x \in |X|$. Notice that if $\varphi^{-1}(x)$ is empty then ${\rm Pr} \{X=x\} = 0$, which means that $\varphi$ is surjective onto $\{x \in |X| \mid {\rm Pr}\{X = x\} \neq 0\}$. The composition of two morphisms is a morphism and the identity map $1_{|X|}: |X| \rightarrow |X|$ gives rise to a morphism of distributions $1_X: X \rightarrow X$ so that the class of finite distributions together with all morphisms between them forms a category denoted $\fd$. As usual, two finite distributions are isomorphic if there exist a morphism between them that has an inverse. If $X$ is a finite distribution, then the slice category (cf. \cite{Awo}) $\fd_X$ of $X$-distributions consists of pairs $(Y, \varphi)$ where $Y$ is a finite distribution and $\varphi: Y \rightarrow X$ is a morphism of finite distributions. 
A morphism of $X$-distributions $f: (Y_1, \varphi_1) \rightarrow (Y_2, \varphi_2)$, consists of a morphism of finite distributions $f: Y_1 \rightarrow Y_2$ such that the following diagram

\begin{center} \begin{tikzcd}
Y_1  \arrow[rr, "f"] \arrow[rd, swap, "\varphi_1"]  & &   Y_2  \arrow[ld, "\varphi_2"] \\
&  X &
\end{tikzcd} \end{center}

\noindent is commutative.

If $x \in |X|$ with ${\rm Pr} \{X = x\} \neq 0$ and $(Y, \varphi)$ is an $X$-distribution then the {\it fiber of $Y$ over $x$} is the finite distribution $Y|_{X=x}$ with support $\varphi^{-1}(x)$ and ${\rm Pr}\{Y|_{X=x} = y\} = \dfrac{{\rm Pr}\{Y=y\}}{{\rm Pr}\{X=x\}}$, for all $y\in \varphi^{-1}(x)$.

If $(Y_1, \varphi_1)$ and $(Y_2, \varphi_2)$ are two $X$-distributions we construct the following product $Y_1 \times_X Y_2$. The support of this distribution is
\begin{equation*}
\vert Y_1 \times_X Y_2 \vert := \{(y_1, y_2)| y_1 \in |Y_1|, y_2 \in |Y_2| \; \text{such that} \; \varphi_1(y_1) = \varphi_2(y_2) \}.
\end{equation*}
If $x=\varphi_1(y_1)=\varphi_2(y_2)$ and ${\rm Pr}\{X =x \} \neq 0$, then

\begin{equation*}
{\rm Pr} \{Y_1 \times_X Y_2 = (y_1, y_2)\} := 
\frac{{\rm Pr} \{Y_1 = y_1\}  \cdot {\rm Pr} \{Y_2 = y_2\}}{{\rm Pr} \{X = x\}}.
\end{equation*}

\noindent Moreover, when ${\rm Pr} \{X = x\}=0$, then 
$$
{\rm Pr} \{Y_1 \times_X Y_2 = (y_1, y_2)\} :=0.
$$

\noindent Finally, the structural morphism of $\psi: \vert Y_1 \times_X Y_2 \vert \rightarrow X$ is
$\psi:= \varphi_1 \circ {\rm pr}_1 = \varphi_2 \circ {\rm pr}_2$, 
where ${\rm pr}_i : \vert Y_1 \times_X Y_2 \vert \rightarrow |Y_i|, i \in \{1, 2\}$ are the usual projections.

We remark that $\vert Y_1 \times_X Y_2 \vert$ is the usual fiber product in the category of sets, but $Y_1 \times_X Y_2$ is not a fiber product in the category $\fd$. However, the distribution $Y_1 \times_X Y_2$ is a product in the following sense. If one constructs the distribution of triples $(x, y_1, y_2)$: $x$ is chosen from $|X|$ according to $X$, $y_1$ and $y_2$ are chosen independently from $\varphi_1^{-1}(x)$ and $\varphi_2^{-1}(x)$
according to $Y_1$ and $Y_2$ respectively, then one obtains a distribution isomorphic to $Y_1 \times_X Y_2$.

Any finite distribution whose support is a one-point set is a final object in $\fd$. We shall denote by $Y_1 \times Y_2$ the product  $Y_1 \times_X Y_2$, where $X$ is any of the final objects of $\fd$. 

Notice that if $Y$ is an $X$-distribution, then the distribution $X \times_X Y$ is isomorphic to $Y$ as $X$-distributions (here we view $X$ as an $X$-distribution via the identity map). 
We will sometimes identify the distribution $X \times_X Y$ with $Y$ without mentioning it, if this is clear from the context. Morally, $X \times_X Y$ is the distribution $Y$ whose associated map $\varphi$ is known.

If $\{X_{\lambda}\}_{\lambda \in \mathbb{N}}$, $\{Y_{\lambda}\}_{\lambda \in \mathbb{N}}$ are ensembles of finite distributions then we define a morphism from the latter to the former as being a set of morphisms of finite distributions $\varphi_{\lambda}: Y_{\lambda} \rightarrow X_{\lambda}$ for all $\lambda$. One can verify immediately that ensembles of finite distributions together with morphisms form a category. If we fix an ensemble $\{X_{\lambda}\}_{\lambda}$, then we obtain the slice category of $\{X_{\lambda}\}_{\lambda}$-ensembles of finite distributions.
In this category we define, as before, the product of the two ensembles $\{Y_{\lambda}\}_{\lambda}$, $\{Z_{\lambda}\}_{\lambda}$ as 
$\{Y_{\lambda} \times_{X_{\lambda}} Z_{\lambda}\}_{\lambda}$.

The first part of the following statement is Definition 2 from \cite{Go90}.

\begin{definition} \label{efsampl} An ensemble $\{X_{\lambda} \}_{\lambda}$ of finite distributions is polynomial-time constructible if there exists a PPT algorithm $A$ such that $A(1^{\lambda})=X_{\lambda}$, for every $\lambda$. An $\{X_{\lambda} \}_{\lambda}$-ensemble of finite distributions $\{(Y_{\lambda}, \varphi_{\lambda}) \}_{\lambda}$ is polynomial-time constructible on fibers if there exist a PPT algorithm $A$, such that for any $x_{\lambda} \in |X_{\lambda}|$ we have $A(1^{\lambda}, x_{\lambda}) = Y_{\lambda}|_{X_{\lambda} = x_{\lambda}}$. 
\end{definition}

We will also use the following notion of computational (or polynomial) indistinguishability from \cite{GM82} and \cite{Go90}.
\begin{definition} Two ensembles of finite distributions $\{X_{\lambda} \}_{\lambda}$ and $\{Y_{\lambda} \}_{\lambda}$ are called computationally indistinguishable if for any PPT distinguisher $D$, the quantity $$\vert \mathrm{Pr}\left\{D(X_{\lambda})=1  \right\} - \mathrm{Pr}\left\{ D(Y_{\lambda})=1 \right\}\vert$$ is negligible as a function of $\lambda$.
\end{definition}

\noindent When referring to ensembles of finite distributions, we will leave out the subscript $\lambda$ if this is clear from the context.

\subsection{Encryption Schemes and Bridges}

A public key (or asymmetric) encryption scheme 
$$
\sS= (\mathrm{KeyGen}_{\sS}, \mathrm{Enc}_{\sS}, \mathrm{Dec}_{\sS})
$$

\noindent is a triple of PPT algorithms as follows:

\begin{itemize}
\item {\bf Key Generation.}  The algorithm $(sk, pk) \leftarrow \mathrm{KeyGen}_{\sS}(1^{\lambda})$ takes a unary representation of the
security parameter $\lambda$ and outputs a secret decryption key $sk$ and a public encryption key $pk$;
\item {\bf Encryption.} The algorithm $c \leftarrow  \mathrm{Enc}_{\sS}(pk, m)$ takes the public key $pk$ and a message $m\in \mathscr{P}$
and outputs a ciphertext $c\in\mathscr{C}$;
\item {\bf Decryption.} The algorithm $m^{\star} \leftarrow \mathrm{Dec}_{\sS}(sk, c)$ takes the secret key $sk$ and a ciphertext $c \in \mathscr{C}$ and outputs
a message $m^{\star} \in \mathscr{P}$;
\end{itemize}
where the finite sets $\mathscr{P}$ and $\mathscr{C}$ represent the plaintext space, respectively the ciphertext space.
The algorithms above must satisfy the correctness property
$${\rm Pr}\left\{\mathrm{Dec}_{\sS}(sk, \mathrm{Enc}_{\sS}(pk, m)) = m \right\} = 1- {\rm negl}(\lambda),$$

\noindent where the probability is taken over the experiment of running the key generation and encryption algorithms and choosing uniformly $m \leftarrow \mathscr{P}$.

A private key (or symmetric) encryption scheme is a public key encryption scheme for which the public and secret keys are equal. 

We say that an instance $pk$ of the public key, or an instance $sk$ of the secret key, is of {\it level} $\lambda_0$ if it is outputted by the key generation algorithm whose input is the unary representation of $\lambda_0$.

\begin{remark}
In the language of ensembles of finite distributions, the public keys of an encryption scheme form an $SK$-ensemble of finite distributions, where $SK$ is the ensemble of secret keys. Moreover, an encryption scheme is just a collection of $PK$-ensembles of finite distributions indexed by the plaintext space that are polynomial-time constructible on fibers (here $PK$ is the ensemble of public keys). 
\end{remark}

 A homomorphic (public-key) encryption scheme 
$$
\sH = (\mathrm{KeyGen}_{\sH}, \mathrm{Enc}_{\sH}, \mathrm{Dec}_{\sH}, \mathrm{Eval}_{\sH})
$$

\noindent is a quadruple of PPT algorithms such that $(\mathrm{KeyGen}_{\sH}, \mathrm{Enc}_{\sH}, \mathrm{Dec}_{\sH})$ is a public-key encryption scheme and the $\mathrm{KeyGen}_{\sH}$ algorithm also outputs an additional evaluation key $evk$ besides $sk$ and $pk$, where the {\bf Homomorphic Evaluation} algorithm 
$\mathrm{Eval}_{\sH}$ takes the evaluation key $evk$,
a circuit $f: \mathscr{P}^{\ell} \rightarrow \mathscr{P}$ and a set of $\ell$ ciphertexts $c_1, ..., c_{\ell} \in \mathscr{C}$, and outputs a ciphertext $c_f$.

We say that a homomorphic encryption scheme $\sH$ is {\it $\mathcal{C}$-homomorphic} for a class of functions $\mathcal{C}=\{\mathcal{C}_{\lambda}\}_{\lambda \in \mathbb{N}}$,
if for any sequence of functions $f_{\lambda} \in \mathcal{C}_{\lambda}$ and respective inputs $\mu_1,...,\mu_{\ell} \in \mathscr{P}$ (where $\ell = \ell(\lambda))$,
it holds that
\begin{equation*}
    \text{Pr} [\mathrm{Dec}_{\sH}(sk,\mathrm{Eval}_{\sH}(evk, f_{\lambda}, c_1, ..., c_{\ell})) \neq f_{\lambda}(\mu_1, ..., \mu_{\ell})] = \text{negl}(\lambda),
\end{equation*}

\noindent where $(pk, sk, evk) \leftarrow \mathrm{KeyGen}_{\sH}(1^{\lambda})$ and $c_i \leftarrow \mathrm{Enc}_{\sH}(pk, \mu_i)$ for all $i$.

 In addition, a homomorphic encryption scheme $\sH$ is {\it compact} if there exist a polynomial $s = s(\lambda)$ such that the output length of $\mathrm{Eval}_{\sH}$ is at most $s$ bits long, regardless of $f$ or the number of inputs.

An encryption scheme is called {\it fully homomorphic (FHE)} if it is homomorphic for the class of all boolean functions and it satisfies the compactness condition.

We now give the definition of a bridge:
\begin{definition}\label{bridge}
Let $\mathscr{S}_j = (\mathscr{P}_j, \mathscr{C}_j, {\rm KeyGen}_j, {\rm Enc}_j, {\rm Dec}_j)$, $j \in \{1, 2\}$ be two encryption schemes. A bridge $\mathbf{B}_{\iota, f}$  from $\mathscr{S}_1$ to $\mathscr{S}_2$ consists of:
\begin{enumerate}
    \item An injective function $\iota: \mathscr{P}_1\to \mathscr{P}_2$ such that:
        \begin{enumerate}
            \item $\iota$ is computable by a deterministic polynomial time algorithm;
            \item there exists a deterministic polynomial time algorithm which computes $\iota^{-1}: \mathscr{P}_2 \to \mathscr{P}_1$, i.e. outputs the symbol $\perp$ if the input is not in the image of $\iota$ and the preimage of the input otherwise,
        \end{enumerate}
    \item A PPT {\it bridge key generation} algorithm, which has the following three stages. First, the algorithm gets the security parameter $\lambda$ and uses it to run the key generation algorithm of $\mathscr{S}_1$ in order to obtain a pair of keys $sk_1, pk_1$. In the second stage the algorithm uses $sk_1$ to find a secret key $sk_2$ of level $\lambda $ for $\mathscr{S}_2$, and then calls the key generation algorithm of $\sS_2$ to produce $pk_2$. In the final stage, the algorithm takes as input the quadruple $(sk_1, pk_1, sk_2, pk_2)$ and outputs a bridge key $bk$.
    
    \item  A PPT algorithm $f$ which takes as input the bridge key $bk$ and a ciphertext $c_1\in \mathscr{C}_1$ and outputs a ciphertext $c_2\in\mathscr{C}_2$, 
\end{enumerate}

\noindent such that 

$$
\text{Pr}\left\{{\rm Dec}_2(sk_2, f(bk, {\rm Enc}_1(pk_1, m)))=\iota(m) \right\}= 1- {\rm negl}(\lambda),
$$

\noindent where the probability is taken over the experiment of running the key generation and encryption algorithms and choosing uniformly $m \leftarrow \mathscr{P}_1$.
\end{definition}

Notice that the definition above includes the case in which any of the two schemes is symmetric. Also, the plaintext spaces are fixed, i.e. they do not depend on the security parameter $\lambda$. One can define a bridge between encryption schemes for which the plaintext spaces do depend on $\lambda$, as in the case of RSA or Paillier cryptosystems. However, in this article we are considering only the former situation.

\begin{remark}
The bridge key generation algorithm does not necessarily output all possible pairs $(sk_1, sk_2)$. Even though any secret key $sk_1$ of the scheme $\mathscr{S}_1$ may be outputted by the key generation algorithm of the bridge, only few $sk_2$'s may occur. The bridge key generation algorithm produces the following $\{SK_{1, \lambda}\}_{\lambda}$-ensembles of finite distributions $\{SK_{2, \lambda} \}_{\lambda}$, $\{PK_{i, \lambda} \}_{\lambda}, i \in \{1, 2\}$, and $\{BK_{\lambda} \}_{\lambda}$. The morphisms between these ensembles of finite distributions are illustrated in Figure \ref{bk:diagram}.

\begin{figure} [ht]
\centering
\begin{tikzcd}
 & BK \arrow[d] & \\
 & PK_1 \times_{SK_1} PK_2 \arrow[ld] \arrow[rd] & \\
 PK_1 \arrow[d] & & PK_2 \arrow[d] \\
 SK_1 &  & SK_2 \arrow[ll]
\end{tikzcd}
\caption{Probability distributions for bridges}
\label{bk:diagram}
\end{figure}
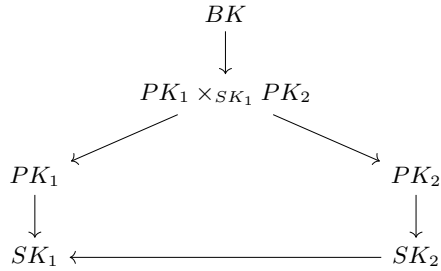
\end{remark}

\noindent We mentioned earlier the idea of thinking of a bridge as a (category theoretical) morphism between encryption schemes. Although we do not claim to have defined a category, from this point of view, it is natural to address the existence of identity morphisms. We briefly explain below that the identity map between one encryption scheme to itself is a bridge.

\medskip

\begin{example} If $\sS$ is an encryption scheme, then the identity map $\mathscr{C} \to \mathscr{C}$ gives rise to a bridge. The bridge key generation algorithm generates a unique secret key $sk$ and two (independently generated) public keys $pk_1, pk_2$ corresponding to this secret key. The algorithm outputs $(sk,pk_1, sk, pk_2, \mathrm{NIL})$. We emphasize that the bridge key and the choices of $pk_1$ and $pk_2$ do not play any role in the evaluation of the bridge map. 
\end{example}

\section{The security of a bridge}

The aim of this section is to define and investigate the IND-CPA security of a bridge. We start by defining and extending the notion of IND-CPA security of a scheme and then we move to the discussion concerning the security of a bridge.

\begin{definition}[IND-CPA Security] \label{IND-CPA} Let $\sS= (\mathrm{KeyGen}_{\sS}, \mathrm{Enc}_{\sS}, \mathrm{Dec}_{\sS})$ be a public key encryption scheme. We define an experiment $\mathrm{Exp}_b[\cA]$ parameterized by a bit $b \in \{0,1 \}$ and an efficient (PPT) adversary $\cA$:

\begin{equation*}
\begin{split}
\mathrm{Exp}_b[\cA](1^{\lambda}): & \text{1. } (pk, sk) \longleftarrow \mathrm{KeyGen}_{\sS}(1^{\lambda}) \\
& \text{2. } (m_0,m_1) \longleftarrow \cA(1^{\lambda}, pk) \\
& \text{3. } \mathrm{ct} \longleftarrow \mathrm{Enc}_{\sS}(pk, x_b) \\
& \text{4. } b' \longleftarrow \cA(\mathrm{ct}) \\
& \text{5. } \mathrm{return}(b')
\end{split}
\end{equation*}

\noindent The advantage of adversary $\cA$ against the IND-CPA security of the scheme is
$$
\mathrm{Adv}^{\text{IND-CPA}}[\cA](\lambda) := \vert \mathrm{Pr} \left\{\mathrm{Exp}_0[\cA](1^{\lambda})=1 \right\} - \mathrm{Pr}\left\{\mathrm{Exp}_1[\cA](1^{\lambda}) = 1 \right\} \vert, 
$$

\noindent where the probability is over the randomness of $\cA$ and of the experiment.
We say that the scheme is IND-CPA secure if for any efficient  adversary $\cA$, the advantage $\mathrm{Adv}^{\text{IND-CPA}}[\cA]$ is negligible as a function of $\lambda$. In the case of a symmetric encryption scheme, the adversary $\cA$ is given access to an encryption oracle.  
\end{definition}

\begin{remark} 
As in the previous definition, when considering the security of a private encryption scheme, it is standard to replace the public key by an encryption oracle. From this point of view, a symmetric encryption scheme is a public encryption scheme whose public key consists of the access to an encryption oracle. Although we will give security definitions and proofs for public key encryption schemes, unless otherwise specified, these can be extended to the symmetric key setting using the above paradigm. 
\end{remark}

Let $\mathscr{S}$ be an encryption scheme and let $K$ be some data outputted by an oracle whose input is the triple $(1^{\lambda}, sk_{\mathscr{S}},  pk_{\mathscr{S}})$. We shall denote by $\mathscr{S}[K]$ the encryption scheme whose public key is the pair $(pk_{\mathscr{S}}, K)$, and the encryption and decryption algorithms are exactly as in $\sS$. The only difference between the schemes $\sS$ and $\sS[K]$ is related to their security. More precisely, an adversary attacking the scheme $\mathscr{S}[K]$ has more information than an adversary attacking $\sS$. We say that an adversary $\cA$ attacking
$\mathscr{S}[K]$ is an adversary attacking {\it $\sS$ with knowledge $K$}.
For example, $K$ can be a set consisting of $\sS$-encryptions of the bit representation of the secret key, as used in \cite{Ge10} for the bootstrapping procedure. It is commonly assumed that such $K$'s do not affect the security of the encryption scheme, assumption called {\it circular security}. The following definition aims to generalize the {\it circular security} assumption for some general data $K$.

\begin{definition}
We say that some knowledge $K$ is negligible for an encryption scheme $\sS$
if for any adversary $\cA$ attacking $\sS[K]$ there exists an adversary $\cA'$ attacking $\sS$ such that 
$$
\vert \mathrm{Adv}^{\text{IND-CPA}}[\cA](\lambda) - \mathrm{Adv}^{\text{IND-CPA}}[\cA'](\lambda)\vert 
$$
is negligible as a function of $\lambda$.
\end{definition}

Notice that any adversary attacking $\sS$ gives rise, in the obvious way, to an adversary attacking $\sS[K]$, so that if $K$ is negligible for $\sS$ then the  IND-CPA security of $\sS$ is equivalent to the IND-CPA security of $\sS[K]$.

In order to define the IND-CPA security of a bridge, we shall associate to it, in a canonical way, an encryption scheme; the security of the bridge will be, by definition, the security of the associated encryption scheme. Let  $\mathbf{B}_{\iota, f}$ be a bridge, then the associated encryption scheme
$$\mathscr{G}_f = (\mathscr{P}_{\sG_f}, \mathscr{C}_{\sG_f}, \mathrm{KeyGen}_{\sG_f}, {\rm Enc}_{\sG_f}, {\rm Dec}_{\sG_f})$$ 
is defined as follows. The plaintext space is $\mathscr{P}_{\sG_f} =\mathscr{P}_1$, and the ciphertext space is $\mathscr{C}_{\sG_f} = \mathscr{C}_1 \times \mathscr{C}_2$. The algorithm ${\rm KeyGen}_{\mathscr{G}_f}$ uses the key generation algorithm of the bridge
to get $sk_1, pk_1, sk_2, pk_2, bk$. The secret key $sk_{\mathscr{G}_f}$ is the pair $(sk_1, sk_2)$, and the public key $pk_{\sG_f}$ is 
$(pk_1, pk_2, bk)$.

For any $m \in \mathscr{P}_{\sG_f}$, its encryption is defined by:
$$
  {\rm Enc}_{\sG_f} (pk_{\sG_f}, m) := \left( a,  f(bk, b) \right),
$$

\noindent where $a, b \leftarrow \mathrm{Enc}_1(pk_1, m)$.
Finally, the decryption of a ciphertext $c_{\sG_f} = (c_1, c_2) \in \mathscr{C}_1 \times \mathscr{C}_2$ is obtained using the formula:

$$
{\rm Dec}_{\sG_f} (sk_{\sG_f}, c_{\sG_f}) := {\rm Dec}_1 (sk_1, c_1).
$$

\noindent We notice that the decryption of $\mathscr{G}_f$ satisfies 

$$
{\rm Dec}_{\sG_f}\Big(sk_{\sG_f}, (a, f(bk,b)) \Big) = \iota^{-1} \Big({\rm Dec}_2(sk_2, f(bk,b)) \Big),
$$

\noindent for any $(a, f(bk,b)) \leftarrow \mathrm{Enc}_{\sG_f}(pk_{\sG_f},m)$ with overwhelming probability, due to the third condition in the definition of a bridge. One can immediately verify that the correctness of the encryption scheme $\mathscr{G}_f$ follows from the correctness of $\sS_1$.

\begin{remark}
The notation and construction are inspired by the construction of the graph of a function.
\end{remark}

Now we define the IND-CPA security of a bridge.

\begin{definition} The IND-CPA security of the bridge $\mathbf{B}_{\iota, f}$ is the IND-CPA security of its
associated encryption scheme $\mathscr{G}_f$ .
\end{definition}

\noindent We have the following immediate result.

\begin{proposition}\label{S1CPA}
If a bridge $\mathbf{B}_{\iota, f}$ is IND-CPA secure, then the encryption scheme $\mathscr{S}_1$ is also IND-CPA secure. 
\end{proposition}

\begin{proof}
Indeed, we can associate to any adversary $\cA_1$ which is trying to break the IND-CPA security of $\mathscr{S}_1$, an adversary $\cA_f$ for the encryption scheme $\mathscr{G}_f$, as follows. For any pair $(a, f(bk,b))$ proposed by the challenger to $\cA_f$, where $a,b \leftarrow \mathrm{Enc}_1(m)$, the attacker $\cA_f$ sends the triple $(\lambda,pk_1, a)$ to $\cA_1$ and returns the output of $\cA_1(\lambda,pk_1,a)$.

It is clear that 

\[
\mathrm{Adv}^{\text{IND-CPA}}[\cA_f](\lambda) = \mathrm{Adv}^{\text{IND-CPA}}[\cA_1](\lambda),
\]

\noindent and the result follows.
\end{proof}

\medskip

In the next theorem, the encryption scheme $\mathscr{S}_1[PK_{\sG_f}]$ is the scheme $\mathscr{S}_1$ with knowledge $PK_{\sG_f}$. Namely, after running the key generation algorithm of $\mathscr{S}_1$ and receiving the pair 
$(sk_1, pk_1)$, the challenger has access to an oracle that runs the second part of the key generation algorithm of the bridge to get $sk_2, pk_2, bk$. Thus, an IND-CPA attacker on this scheme will receive $pk_1, pk_2, bk$.

\begin{theorem} \label{knowledge:thm} The encryption scheme $\mathscr{S}_1[PK_{\sG_f}]$ is IND-CPA secure if and only if $\sG_f$ is IND-CPA secure.
\end{theorem}

\begin{proof}
We first show that if $\sG_f$ is IND-CPA secure, then $\sS_1[PK_{\sG_f}]$ is IND-CPA secure. Suppose $\cA$ is an IND-CPA attacker on $\sS_1[PK_{\sG_f}]$ scheme. We construct the following adversary $\cB$ attacking the IND-CPA security of $\sG_f$ as follows. At start, $\cB$ takes as input $(1^{\lambda}, pk_{\sG_f})$ and executes the program $\cA(1^{\lambda}, pk_{\sG_f})$. The attacker $\cB$ receives $(m_0,m_1) \leftarrow \cA(1^{\lambda}, pk_{\sG_f})$ and sends this pair to its challenger. The latter samples $b \leftarrow \{0,1 \}$ and returns to $\cB$ the challenge $c=(c_1,f(bk,c_1'))$, where $c_1, c_1' \leftarrow \mathrm{Enc}_1(pk_1,m_b)$.
Finally, $\cB$ terminates by outputting the bit $b'\leftarrow \cA(c_1)$. 
One obtains that
$$\mathrm{Adv}^{\mathrm{IND-CPA}}_{\sG_f}[\cB]({\lambda})= \mathrm{Adv}^{\mathrm{IND-CPA}}_{\sS_1[PK_{\sG_f}]}[\cA]({\lambda}),$$
which proves this implication.

To prove the other implication, we first point out that using a standard hybrid argument one can show that the IND-CPA security of an encryption scheme is equivalent to its $2$-IND-CPA security (see \cite{Sm16} for a detailed discussion). As opposed to the IND-CPA game, in the $2$-IND-CPA game the attacker receives from the challenger two encryptions of $m_b$, instead of one.

Suppose that $\cB$ is an IND-CPA attacker on $\sG_f$. We construct a $2$-IND-CPA attacker $\cA$ for the scheme $\sS_1[PK_{\sG_f}]$ as follows.
The attacker $\cA$ receives as input $(1^{\lambda},pk_{\sG_f})$ and sends this to $\cB$. On this input, the attacker $\cB$ produces two messages $m_0,m_1 \in \sP_1$ which are sent to $\cA$ and the latter passes them to its challenger. After receiving $m_0,m_1$, the challenger of $\cA$ chooses $b \leftarrow \{0,1\}$ and returns $c_1, c_1' \leftarrow \mathrm{Enc}_1(m_b)$ to the attacker $\cA$. The attacker $\cA$, knowing $bk$, is able to compute $f(bk,c_1') \in \cC_2$ and finishes by outputting $b' \leftarrow \cB(c_1, f(bk,c_1'))$. Now, one can verify that
$$\mathrm{Adv}^{\mathrm{2-IND-CPA}}_{\sS_1[PK_{\sG_f}]}[\cA]({\lambda}) = \mathrm{Adv}^{\mathrm{IND-CPA}}_{\sG_f}[\cB]({\lambda}).$$
By the discussion in the previous paragraph, the scheme $\sS_1[PK_{\sG_f}]$ is $2$-IND-CPA secure, so that $\cA$ has negligible advantage. The last equality shows that $\cB$ has also negligible advantage, which ends the argument. \end{proof}

Recall that the \textit{bridge key generation algorithm} produces the following ensembles of $\{SK_{1,\lambda} \}_{\lambda}$ distributions: $\{PK_{1,\lambda} \}_{\lambda}$, $\{PK_{2, \lambda} \}_{\lambda}$ and $\{BK_{\lambda}\}_{\lambda}$. Let $\mathscr{F}$ be the ensemble of finite distributions of triples $(pk_1, pk_2, bk)$. Note that $\pi_1 : \mathscr{F} \to PK_1$ is a morphism of finite distributions, so $\mathscr{F}$ is a $PK_1$-distribution as discussed in Section \ref{finitedist}.

\begin{theorem} \label{mt}
Assume that $\sS_1$ is IND-CPA secure and there exists a polynomial time constructible on fibers ensemble of $PK_1-$distributions $\widetilde\sF$ which is computational indistinguishable from $\sF$. Then the bridge $\mathbf{B}_{\iota, f}$ is IND-CPA secure.
\end{theorem}

\begin{proof} Without losing generality we assume that $\mathscr{P}_1 = \{0, 1\}$. By the above theorem, it is enough to prove that $\sS_1[PK_{\sG_f}]$ is IND-CPA secure. We do the proof by contradiction, so we suppose that $\cA$ is an adversary attacking the scheme $\sS_1[PK_{\sG_f}]$ with non-negligible advantage. We think of
$\cA$ as being a distinguisher between the ensembles of distributions $\mathscr{F} \times_{PK_1} {\rm Enc}_1(PK_1,0)$ and $\mathscr{F} \times_{PK_1} {\rm Enc}_1(PK_1,1)$. The first claim is that, if $\cA$ can distinguish with non-negligible advantage between these two distributions then $\cA$ distinguishes with non-negligible advantage between $\widetilde\sF \times_{PK_1} {\rm Enc}_1(PK_1,0)$ and $\widetilde\sF \times_{PK_1} {\rm Enc}_1(PK_1,1)$. 
To prove the claim we suppose that this is not the case and we construct a distinguisher $\cD$ for the distributions
$\sF$ and $\widetilde\sF$. As the ensemble of distributions $\widetilde\sF$ is computationally indistinguishable from $\sF$, for every $\lambda$, the distribution $\widetilde\sF_{\lambda}$ consists of triples of the form $(pk_1,\alpha,\beta)$. 

The distinguisher $\cD$ runs as follows. It first receives a triple $(pk_1, x, y)$ from the challenger, chooses at random a bit $b \leftarrow \{0, 1\}$ and encrypts $b$ using $pk_1$ to obtain a ciphertext $c$. The distinguisher $\cD$ sends the quadruple $(pk_1, x, y, c)$ to $\cA$ and outputs

\begin{equation*}
\cD(pk_1, x, y) := 
\begin{cases}
1 & \text{if} \; \; \cA (pk_1, x, y, c) = b \\
0 & \text{otherwise} \\
\end{cases}.
\end{equation*}

\noindent We note that the labels $b=1$ and $b=0$, as outputted by $\cA$, correspond to the ensembles $\sF$ and $\widetilde\sF$, respectively. Notice that
\begin{equation*}
\mathrm{Pr} \left\{\mathrm{Exp}_1[\cD] = 1 \right\} = 
\frac{1}{2} \mathrm{Pr} \left\{\mathrm{Exp}_0[\cA \vert_{\sF}] = 0 \right\} + \frac{1}{2} \mathrm{Pr} \left\{\mathrm{Exp}_1[\cA \vert_{\sF}] = 1 \right\}, 
\end{equation*}

\noindent where $\mathrm{Exp}_b[\cA \vert_{\sF}]$ means that in
the experiment $\mathrm{Exp}_b$ the challenger chooses the triple
$(pk_1,x,y)=(pk_1, pk_2, bk)$ according to $\sF$. Using analogous notation for $\widetilde\sF$, we have:
\begin{equation*}
\mathrm{Pr} \left\{\mathrm{Exp}_0[\cD] = 1 \right\} = 
\frac{1}{2} \mathrm{Pr} \left\{\mathrm{Exp}_0[\cA \vert_{\widetilde \sF}] = 1 \right\} + \frac{1}{2} \mathrm{Pr} \left\{\mathrm{Exp}_1[\cA \vert_{\widetilde \sF}] = 0 \right\}. 
\end{equation*}

\noindent Since the advantage of $\cA|_{\sF}$ is non-negligible, there exists a positive integer $k$ such that 

\begin{equation} \label{eq1}
  \left| \mathrm{Pr} \left\{\mathrm{Exp}_1[\cD] = 1 \right\} - \dfrac{1}{2} \right| > \dfrac{1}{\lambda^k}
\end{equation}

\noindent for infinitely many $\lambda$'s. Also, since $\mathrm{Adv}[\cA\vert_{\widetilde \sF}](\lambda) = \mathrm{negl}(\lambda)$, we have 

\begin{equation} \label{eq2}
  \left| \mathrm{Pr} \left\{\mathrm{Exp}_0[\cD] = 1 \right\} - \dfrac{1}{2} \right| = \mathrm{negl}(\lambda).
\end{equation}

\noindent From (\ref{eq1}) and (\ref{eq2}) we infer that 
$$\mathrm{Adv}[\cD](\lambda) = \left|  \mathrm{Pr} \left\{\mathrm{Exp}_1[\cD] = 1 \right\} -  \mathrm{Pr} \left\{\mathrm{Exp}_0[\cD] = 1 \right\} \right|$$ 

\noindent is non-negligible, which contradicts the assumption about the computational indistinguishability of the two distributions $\sF$ and $\widetilde \sF$.

Now we use $\cA\vert_{\widetilde \sF}$ to construct an adversary $\cB$
on $\sS_1$. After receiving the pair $(pk_1, c)$ (as before $c \leftarrow {\rm Enc}_1(pk_1, b)$) from the challenger, $\cB$ is 
using the sampling algorithm of $\widetilde \sF$ to get a triple
$(pk_1, \alpha, \beta)$. The adversary $\cB$ sends $(pk_1, \alpha, \beta, c)$ 
to $\cA\vert_{\widetilde \sF}$ and outputs the bit received from
it. It is clear that

\begin{equation*}
  \mathrm{Adv}[\cB](\lambda) =  \mathrm{Adv}[\cA\vert_{\widetilde \sF}](\lambda)
\end{equation*}

\noindent so that $\cB$ breaks the IND-CPA security of $\sS_1$ with
non-negligible advantage, and this contradicts our assumption.
\end{proof}

\medskip

\section{A general recipe for constructing bridges} 
\label{recipe:gentry}

As we shall explain in what follows, the \textit{Recrypt} algorithm, used in the bootsrapping procedure that transforms a somewhat homomorphic encryption scheme into a fully homomorphic encryption scheme (see \cite{Ge10}), can be adapted to our situation in order to give a general recipe for the construction of a bridge. We will call this method \textit{Gentry's recipe} and say that the bridges obtained using it are of \textit{Gentry type}.

Let us consider an encryption scheme $$\sS=(\sP_{\sS}, \sC_{\sS},\mathrm{KeyGen}_{\sS}, \mathrm{Enc}_{\sS}, \mathrm{Dec}_{\sS})$$ and a homomorphic encryption scheme $$\sH=(\sP_{\sH}, \sC_{\sH}, \mathrm{KeyGen}_{\sH}, \mathrm{Enc}_{\sH}, \mathrm{Dec}_{\sH}, \mathrm{Eval}_{\sH}),$$
such that $\sP_{\sH}$ has a ring structure and there exists an injective map $\iota : \sP_{\sS} \hookrightarrow \sP_{\sH}$ satisfying the properties 1.(a)-(b) in Definition \ref{bridge}.

In this construction, the key generation algorithm is as follows. First, it runs $\mathrm{KeyGen}_{\sS}(1^{\lambda})$ to sample from the distribution $SK_{\sS}$ and then, independently, it runs $\mathrm{KeyGen}_{\sH}(1^\lambda)$ to sample from $SK_{\sH}$. We point out that the distribution $SK_2$ in the definition of the bridge is in fact the product $SK_{\sS} \times SK_{\sH}$ and the map $SK_2 \to SK_1$ (see Figure \ref{bk:diagram}) is the projection on the first component $SK_{\sS} \times SK_{\sH} \to SK_{\sS}$. Samples for the public keys $pk_{\sS}$ and $pk_{\sH}$ are generated, independently, using the key generation algorithms of the two schemes. Given a quadruple $(sk_{\sS}, pk_{\sS}, sk_{\sH}, pk_{\sH})$ constructed as above, the algorithm creates $bk$ as the vector of encryptions of all the bits of $sk_{\sS}$ under $pk_{\sH}$ (see below). This is how the distribution of bridge keys $BK$ is obtained.

The PPT algorithm $f$ mentioned in the third part of Definition \ref{bridge} is in this case the homomorphic evaluation (in $\sH$) of the algorithm $\mathrm{Dec}_{\sS}$. 
 We need to realise $\mathrm{Dec}_{\sS}$ as a map $\mathscr{P}_{\mathscr{H}}^{\ell} \rightarrow \mathscr{P}_{\mathscr{H}}$, and for this we use the ring structure on $\sP_{\sH}$. 
Suppose that the ciphertext space $\sC_{\sS}$ is a subset of $\{ 0, 1\}^n$ and that the set of secret keys is a subset of $\{ 0, 1\}^e$,
so that ${\rm Dec}_{\sS}: \{ 0, 1\}^e \times \{ 0, 1\}^n \rightarrow \mathscr{P}_{\sS}$. We construct the map $\widetilde{\rm Dec}_{\sS}: \mathscr{P}_{\mathscr{H}}^e \times \mathscr{P}_{\mathscr{H}}^n \rightarrow \mathscr{P}_{\mathscr{H}}$ as follows. Letting $\mathscr{P}_{\mathscr{H}}$ be a subset of $\{0, 1\}^m$, we have that $\iota \circ {\rm Dec}_{\sS}: \{ 0, 1\}^e \times \{ 0, 1\}^n \rightarrow \mathscr{P}_{\mathscr{H}}$ is a vector $(g_1,...,g_m)$ of boolean circuits expressed using ${\rm XOR}$ and ${\rm AND}$ gates. Let $\tilde{g_i}: \mathscr{P}_{\mathscr{H}}^e \times \mathscr{P}_{\mathscr{H}}^n \rightarrow \mathscr{P}_{\mathscr{H}}$ be the circuit obtained by replacing each ${\rm XOR}(x,y)$- gate by $x\oplus y:=2(x+y) - (x+y)^2$ and each ${\rm AND}(x,y)$ gate by $x\otimes y := x \cdot y$, where $+$ and $\cdot$ are the addition and multiplication in $\mathscr{P}_{\mathscr{H}}$. Notice that the subset of $\mathscr{P}_{\mathscr{H}}$ consisting of its zero element $0_\sH$ and its unit $1_\sH$ together with $\oplus$ and $\otimes$ is a realisation of the field with two elements inside $\sP_\sH$. In other words, if $c=(c[1],...,c[n]) \in \sC_{\sS}$ and $sk_\sS=(sk[1],...,sk[e])$ is the secret key, then $\tilde{g_i}(sk[1]_\sH,...,sk[e]_\sH, c[1]_\sH,...,c[n]_\sH) = m_\sH$ if $g_i(sk[1],...,sk[e],c[1],...,c[n]) = m$ for all $i$, where $m \in \{0, 1\}$.
For an element $x \in \sP_\sH$, we let $[x=1_{\sH}]$ be the equality test, which returns $1$ if $x=1_{\sH}$ and $0$ otherwise. 
Finally, $\widetilde{\rm Dec}_{\sS}:\mathscr{P}_{\mathscr{H}}^e \times \mathscr{P}_{\mathscr{H}}^n \to \{0,1\}^m$ is defined by: $$\left(\left[ \tilde{g_i}(y_1,..., y_e, x_1,..., x_n) = 1_\sH \right]\right)_{i = \overline{1, m}}.$$ 
One can verify that 
$$\widetilde{\rm Dec}_{\sS}(sk[1]_\sH,...,sk[e]_\sH, c[1]_\sH,...,c[n]_\sH)= \iota\circ {\rm Dec}_{\sS}(sk,c).$$

Now we are ready to define the bridge map. Given a ciphertext $c \in \mathscr{C}_{\sS}$, the algorithm $f$ first encrypts the $n$ bits of $c$ (viewed as elements of $\sP_\sH$) under $pk_{\sH}$ and retains these encryptions in a vector $\tilde{c}$. The bridge key $bk$ is obtained by encrypting the bits of $sk_{\sS}$ under $pk_{\sH}$. Then, the algorithm outputs:
$$f(bk,c) = \mathrm{Eval}_{\sH}(evk_{\sH},\widetilde{\rm Dec}_{\sS},bk, \tilde c)$$
Assuming that $\sH$ can evaluate $\widetilde{\rm Dec}_{\sS}$ we have:
\begin{eqnarray*}
 {\rm Dec}_{\sH} (f(bk,c)) & = & {\rm Dec}_{\sH}  \left( \mathrm{Eval}_{\sH}(evk_{\sH},\widetilde{\rm Dec}_{\sS}, bk, \tilde c) \right)  \\
& = &  \iota \left( {\rm Dec}_{\sS} ({\rm Dec}_{\sH}(bk), {\rm Dec}_{\sH}(\tilde{c})) \right)\\
& = & \iota \left( {\rm Dec}_{\sS} (sk_{\sS}, c) \right)
\end{eqnarray*}
which shows that third condition in the definition of a bridge is satisfied.

\begin{remark}
The above construction relies on the fact that the plaintext space of $\sH$, being a ring, can be used to simulate an $\mathbb{F}_2$-structure inside it. 
\end{remark}

An example of the above construction can be found in \cite{GHS12}, where the authors managed to homomorphically evaluate the AES-128 circuit (encryption and decryption) using an optimized implementation of the BGV scheme \cite{BGV12}. Once the plaintext spaces and the embedding $\iota$ are fixed, the evaluation of this decryption circuit can be seen as a Gentry type bridge. The bridge key consists of the BGV encryptions of the eleven AES round keys (see Section 4 of \cite{GHS12}). We note that here the round keys are embedded in the plaintext, so it was not necessary to encrypt the bits of the round keys, as discussed at the beginning of the section. This results in a simpler homomorphic evaluation of AES decryption. Nonetheless, this bridge is essentially obtained using Gentry's recipe.

\subsection{On the security of Gentry type bridges}

The aim of this subsection is to show that if $\sS$ and $\sH$ are IND-CPA secure, then any Gentry type bridge $B_{\iota,f}$ from $\sS$ to $\sH$ is IND-CPA secure. The plan is to apply Theorem \ref{mt} to the above construction. 

Recall that $\mathscr{F}$ is the ensemble of finite distributions of triples $(pk_\sS, pk_\sH, bk)$, where $bk$ is a vector of encryptions of the form $(bk[1],...,bk[e])$ with $bk[i] \leftarrow {\rm Enc}_\sH(pk_{\sH}, sk[i]_\sS)$ for all $i$. Next, let $\widetilde{\mathscr{F}}$ be the ensemble of finite distributions of triples $(pk_\sS, pk_\sH, \widetilde{bk})$, where $pk_{\sS}$, $pk_{\sH}$ are independently outputted by $\mathrm{KeyGen_{\sS}}$ and $\mathrm{KeyGen}_{\sH}$, respectively and $\widetilde{bk}:=(\widetilde{bk}[1],...,\widetilde{bk}[e])$ with $\widetilde{bk}[i] \leftarrow {\rm Enc}(pk_\sH, 0_\sH)$ for all $i\in \overline{1,e}$. Notice that $\widetilde{\mathscr{F}}$ is polynomial-time constructible on fibers as a $PK_\sS$-ensemble of finite distributions (see Definition \ref{efsampl}).
Let us remark that one can choose $\widetilde{\mathscr{F}}$ in a different way, setting $\widetilde{bk}$ to be a  vector of encryptions of any fixed $e$-long bit vector. If the scheme $\sH$ is IND-CPA secure, then one can prove by a standard hybrid argument (see the next proposition) that the two versions are in fact computational indistinguishable. Therefore, the choice of the particular fixed bit vector that is encrypted to get $\widetilde{bk}$ does not matter.

\begin{proposition} \label{prop:proof}
If $\sH$ is IND-CPA secure, then the ensembles $\mathscr{F}$ and $\widetilde{\mathscr{F}}$ are computationally indistinguishable. 
\end{proposition}

\begin{proof} 
Let $\cD$ be a distinguisher between the two ensembles $\mathscr{F}$ and $\widetilde{\mathscr{F}}$. We denote by $\sG_i$ the distribution of triples $(pk_{\sS}, pk_{\sH}, x)$ where the pair $(pk_{\sS}, pk_{\sH})$ is chosen exactly as in the case of $\mathscr{F}$, or  $\widetilde{\mathscr{F}}$, and $x:=(x[1],...,x[e])$ where $x[j] \leftarrow {\rm Enc}(pk_\sH, sk_{\sS}[j])$ for all $j \in \overline{1,i}$ 
and $x[j] \leftarrow {\rm Enc}(pk_\sH, 0)$ for all $j \in \overline{i+1,e}$. Notice that $\{\sG_{e(\lambda)}\}_{\lambda}$ is the same as $\mathscr{F}$, and $\{\sG_0\}_{\lambda}$ is $\widetilde{\mathscr{F}}$. For each $i\in \overline{1,e}$ we construct an attacker $\cB_i$ on the scheme $\sH$ as follows. The attacker receives from the challenger the triple $(1^{\lambda}, pk_{\sH}, c)$,
where $c$ is either an encryption of 0 or an encryption of 1. The attacker uses ${\rm KeyGen}_{\sS}$ to generate a pair $(sk_{\sS}, pk_{\sS})$ and then constructs an $e$-long vector $y$ as follows:
$y[j] \leftarrow {\rm Enc}(pk_{\sH}, sk_{\sS}[j])$ for $j<i$, $y[i] =c$, and $y[j] \leftarrow {\rm Enc}(pk_{\sH}, 0)$ for $j>i$. Then the attacker $\cB_i$ runs $\cD(1^{\lambda}, pk_{\sS}, pk_{\sH}, y)$ and outputs $sk[i]$ if the answer received from $\cD$ is $\mathscr{F}$ and $0$ otherwise. Basically, $\cD$ can be used as a distinguisher between the ensembles $\{\sG_{i-1}\}_{\lambda}$ and $\{\sG_i\}_{\lambda}$, which gives rise to $\cB_i$. Notice that 

$$
\mathrm{Adv}^{\mathrm{IND-CPA}}[\cD](\lambda) \leq \sum_{i=1}^{e(\lambda)} \mathrm{Adv}^{\mathrm{IND-CPA}}[\cB_i](\lambda),
$$

\noindent where we used the fact that the advantage of $\cB_i$ is equal to the advantage of $\cD$ as a distinguisher between $\sG_i$ and $\sG_{i-1}$. Since $\sH$ is IND-CPA secure and $e(\lambda)$ is polynomial in $\lambda$, we get that $\cD$ has negligible advantage. 
\end{proof}

\noindent The result of Proposition \ref{prop:proof} combined with Theorem \ref{mt} yields the following result:

\begin{theorem}
Assume that $\sS$ and $\sH$ are both IND-CPA secure, then any Gentry type bridge $\mathbf{B}_{\iota, f}$ from $\sS$ to $\sH$ is IND-CPA secure.
\end{theorem}

\subsection{A variant of Gentry's recipe}
The aim of this subsection is to give a new variant of Gentry's recipe for the construction of bridges. For this, we need first to introduce the product of two encryption schemes. Suppose that $\sS_i=(\sP_i, \sC_i,\mathrm{KeyGen}_{i}, \mathrm{Enc}_{i}, \mathrm{Dec}_{i})$, $i \in \{1, 2\}$ are two encryption schemes, then the product $\sS_1 \times \sS_2$ is defined as follows. The plaintext space is defined as $\sP_1 \times \sP_2$ and the ciphertext space as $\sC_1 \times \sC_2$.
The Key Generation algorithm of the product scheme uses independently the key generation algorithms of the two schemes to produce two pairs $(sk_1, pk_1)$ and $(sk_2, pk_2)$ of keys and sets the secret key as $(sk_1, sk_2)$, and sets the public key as $(pk_1, pk_2)$. An encryption of a message $(m_1, m_2) \in \sP_1 \times \sP_2$ is just a pair $(c_1, c_2)$, where $c_1 \leftarrow {\rm Enc}_1(pk_1, m_1)$ and $c_2 \leftarrow {\rm Enc}_2(pk_2, m_2)$. Finally, the decryption of $(c_1, c_2)$ is $({\rm Dec}_1(sk_1, c_1), {\rm Dec}_2(sk_2, c_2))$.
In the same way, one can define the product of $p \geq 2$ encryption schemes. If $\sH$ is an encryption scheme, we shall denote by ${\sH}^p$ the product of $p$ copies of $\sH$. 

Now, we describe this new construction. We use the same notations as in the beginning of this section, and we assume that $\sP_{\sH} = \{0, 1\}$. Let $\iota: \sP_{\sS}  \hookrightarrow \{0, 1\}^p$ be a representation of the plaintext space of $\sS$, which can be viewed as the map $\iota: \sP_{\sS}  \hookrightarrow \sP^p_\sH$, by identifying $\{0, 1\}^p$ with the plaintext space of $\sH^p$. We construct a bridge from $\sS$ to $\sH^p$ as follows. Notice that the decryption algorithm of $\sS$ is in fact a $p$-long vector of boolean algorithms $g_i: \{0, 1\}^e \times \{0, 1\}^n \rightarrow \{0, 1\}$, that is
${\rm Dec}_{\sS} (sk_{\sS}, c) = (g_1(sk_{\sS}, c),...,g_p(sk_{\sS}, c))$, where $\{0, 1\}^n$ and $\{0,1\}^{e}$ correspond to $\sC_{\sS}$ and the support of secret keys of $\sS$, respectively. The bridge key $bk$ is obtained by encrypting the bits of $sk_{\sS}$ under each component of the public key of $\sH^p$.

The bridge map $f$ is the vector obtained by homomorphically evaluating the circuits $g_i$ in $\sH$. More precisely
$$
f(bk,c) = \left( \mathrm{Eval}_{\sH}(evk_{\sH}, g_i, bk, \tilde c) \right)_{i=\overline{1,p}},
$$
where $\tilde c$ is defined as above. Notice that, if $c \leftarrow {\rm Enc}(pk_{\sS}, m)$ then 

$$
{\rm Dec}_{\sH^p}(f(bk,c)) = \iota(m).
$$

\medskip 

\noindent\textbf{Security}. As in the previous subsection, it can be shown that if $\sS$ and $\sH$ are IND-CPA secure, then the bridge is also IND-CPA secure. The proof is very similar to that of Proposition \ref{prop:proof}, hence omitted here. 

\section{Conclusions}

Access to secure and efficient bridges between homomorphic encryption schemes would be helpful for applications of cloud computing on sensitive data. Investigating theoretical results for proving the security of such protocols is therefore an important topic. Our main theorem is such a tool, and a particular case of it recovers the already known security of Gentry-type bridges.

\subsection*{Acknowledgements}

The authors are indebted to George Gugulea and Mihai Togan for helpful discussions and comments during the preparation of this work. We are also grateful to the anonymous reviewers for useful suggestions.

\bibliographystyle{MathPhySci}

\begin{thebibliography}{00}

\bibitem{ASP14} Alperin-Sheriff J., Peikert C.: Faster Bootstrapping with Polynomial Error. In: Garay J.A., Gennaro R. (eds) Advances in Cryptology, CRYPTO 2014, LNCS, vol 8616, pp. 297--314. Springer, Berlin, Heidelberg (2014). 

\bibitem{Awo} Awodey, S.:  Category theory, 2nd edn, Oxford University Press, Oxford (2010).

\bibitem{BPP18} Barcau, M.,  Pa\c sol, V., Ple\c{s}ca, C.: Monoidal Encryption over $\mathbb{F}_2$, In: Lanet JL., Toma C. (eds) Innovative Security Solutions for Information Technology and Communications, SECITC 2018, LNCS, vol 11359, pp. 504--517, Springer, Cham (2019). 

\bibitem{BGGJ} Boura, C., Gama, N., Georgieva, M., Jetchev, D.: CHIMERA: Combining Ring-LWE-based Fully Homomorphic Encryption Schemes, Journal of Mathematical Cryptology \textbf{14}(1), pp. 316 -- 338 (2020).  

\bibitem{BGV12} Brakerski, Z., Gentry, C., Vaikuntanathan, V.: (Leveled) fully homomorphic encryption without bootstrapping, ACM Transactions on Computation Theory \textbf{6}(3), No. 13, pp. 1--36 (2014). 

\bibitem{BLPRS13} Brakerski, Z., Langlois, A., Peikert, C., Regev, O., Sthel\' e: 
Classical hardness of learning with errors, In: STOC '13: Proceedings of the forty-fifth annual ACM symposium on Theory of Computing, pp. 575--584.  


\bibitem{B12} Brakerski, Z.: Fully homomorphic encryption without modulus switching from classical GapSVP, In: Safavi-Naini R., Canetti R. (eds) Advances in Cryptology, CRYPTO 2012, LNCS, vol. 7417, pp. 868 -- 886. Springer, Berlin, Heidelberg (2012).  

\bibitem{CIL17} Castagnos, G., Imbert, L., Laguillaumie, F.:  Encryption Switching Protocols Revisited: Switching Modulo $p$, In: Katz J., Shacham H. (eds) Advances in Cryptology, CRYPTO 2017, LNCS, vol. 10401, pp. 255 -- 287, Springer, Cham (2017). 

\bibitem{CSGN20}certSIGN RD: CSGN GitHub repository, \url{https://github.com/certFHE/CSGN}, Last accessed on 20 May 2021.

\bibitem{CGGI19} Chillotti, I., Gama, N., Georgieva, M. and Izabachène, M.: TFHE: Fast Fully Homomorphic Encryptionover the Torus, Journal of Cryptology, \textbf{33} pp. 34 -- 91 (2020).

\bibitem{Coh19} Cohen, A.: What About Bob? The Inadequacy of CPA Security for Proxy Reencryption, PKC (2) 2019, pp. 287--316.

\bibitem{CPP16} Couteau, G., Peters, T., Pointcheval, D.: Encryption Switching Protocols, In: Robshaw M., Katz J. (eds) Advances in Cryptology, CRYPTO 2016, LNCS, vol. 9814, pp. 308 -- 338. Springer, Berlin, Heidelberg (2016).

\bibitem{Rasta18} Dobraunig, C., Eichlseder, M., Grassi, L., Lallemand, V., Leander, G., List, E., Mendel, F., Rechberger, C.: Rasta: A Cipher with Low ANDdepth and Few ANDs per Bit, In: Advances in Cryptology, CRYPTO 2018, LNCS, vol. 10991, pp. 662 -- 692. Springer, Cham (2018).

\bibitem{Pasta21} Dobraunig, C., Grassi, L., Helminger, L., Rechberger, C., Schofnegger, M. and Walch, R.: Pasta: A Case for Hybrid Homomorphic Encryption. In Cryptology ePrint Archive (2021).

\bibitem{Do03} Dodis, Y., Ivan, A.: Proxy cryptography revisited, In: Proceedings of the Tenth Network
and Distributed System Security Symposium, February 2003.


\bibitem{DN21} Dottling, N., Nishimaki, R.: Universal Proxy Re-Encryption, Cryptology ePrint Archive, Report 2018/840, to appear in PKC '21.

\bibitem{FV12} Fan, J., Vercauteren, F.: Somewhat Practical Fully Homomorphic Encryption, IACR Cryptol. ePrint Arch. 2012: 144 (2012).

\bibitem{Ge09T} Gentry, C: A fully homomorphic encryption scheme, PhD thesis, Stanford University, 2009.

\bibitem{Ge10} Gentry, C.: Computing arbitrary functions of encrypted data, Communications of the ACM, \textbf{53}(3), pp.  97 -- 105 (2010).

\bibitem{GHS12} Gentry C., Halevi S., Smart N.P.: Homomorphic Evaluation of the AES Circuit. In: Safavi-Naini R., Canetti R. (eds) Advances in Cryptology, CRYPTO 2012, LNCS, vol. 7417, pp. 850--867. Springer, Berlin, Heidelberg (2012).

\bibitem{Go90} Goldreich, O.: A note on computational indistinguishability, Information Processing Letters, \textbf{34}(6), pp. 277 -- 281.

\bibitem{GM82} Goldwasser, S., Micali, S.: Probabilistic encryption and how to play mental poker keeping secret all partial information, In: STOC '82: Proceedings of the fourteenth annual ACM symposium on Theory of computing, pp. 365 -- 377. Association for Computing Machinery, New York, NY (1982).

\bibitem{GM84} Goldwasser, S., Micali, S.: Probabilistic Encryption, Journal of Computer and System Sciences, \textbf{28}(2), pp. 270--299 (1984).

\bibitem{HELIB} HElib library homepage: An Implementation of homomorphic encryption by Halevi and Shoup, https://github.com/shaih/HElib/.

\bibitem{Mic18} Micciancio, D.: On the Hardness of Learning With Errors with Binary Secrets, Theory of Computing, \textbf{14}(13), pp. 1--17 (2018).

\bibitem{Reg05} Regev, O.: On Lattices, Learning with Errors, Random Linear Codes, and Cryptography, In Harold N. Gabow and Ronald Fagin, editors, STOC, pages 84–93. ACM, 2005.


\bibitem{SEAL} Microsoft Research, Redmond, WA., Microsoft SEAL (release 3.6), \url{https://github.com/Microsoft/SEAL}, November, 2020.

\bibitem{SYY} Sander, T., Young, A., Yung, M.: Non-Interactive CryptoComputing For $NC^1$. In: FOCS '99: Proceedings of the 40th Annual Symposium on Foundations of Computer Science, pp. 554 -- 566, IEEE Computer Society, NW Washington, DC, United States (1999).

\bibitem{Sm16} Smart, N.: Cryptography Made Simple, Springer, Cham (2016).
\end{thebibliography}

\appendix

\section{Examples of Gentry bridges}

The aim of this appendix is to emphasize the fact that, for an encryption scheme $\sS$, different representations for the decryption algorithm $\mathrm{Dec}_{\sS}$ give rise to different bridges from $\sS$ to a FHE scheme $\sH$. For practical applications, one can select the appropriate representation that best suits the implementation of the desired application. Having this in mind, we chose to exhibit the encryption scheme CSGN introduced in \cite{BPP18} and implemented in \cite{CSGN20}, whose decryption algorithm admits at least four fundamentally different representations. We shall restrict ourselves in discussing the security of these bridges, because the security of the CSGN scheme is not entirely understood. 

\subsection{Description of the CSGN scheme}

We give a brief description of the CSGN scheme. For more details regarding the parameter selection, we refer to \cite{BPP18}. The plaintext space is the field $\F_2$ and the ciphertext space of this scheme is $\F_2^n$ with the monoid structure defined by component-wise multiplication. A simplified version of the scheme is defined as follows.

\begin{itemize}
	\item $\mathrm{KeyGen}_{\rm CSGN}(1^{\lambda})$: Choose dimension parameters $n$, $d$ and $s$ of size $\mathrm{poly}(\lambda)$, a uniformly random subset $S$ of $\{1,2, \dots, n \}$ of size $s$, and a finite distribution $X$ on $\{1, 2, ..., d\}$ according to \cite{BPP18}. Set the secret key $sk$ to be the characteristic function of $S$, viewed as a bit vector.
	
	\item $\mathrm{Enc}_{\rm CSGN}$: To encrypt $0$, choose first $k \in \{1, 2, ..., d\}$ according to $X$ and then choose uniformly at random $d$ numbers  $i_1, \dots, i_d$ from the set  $\{1, 2, \dots, n \}$, such that exactly $k$ of them are in $S$. Finally, output the vector in $\F_2^n$ whose components corresponding to the indices $i_1$, $\dots$, $i_d$ are equal to $0$ and the others are equal to $1$. To encrypt $1$, choose uniformly at random $d$ numbers  $i_1, \dots, i_d$ from the set  $\{1, 2, \dots, n \}$, such that none of them is in $S$, and output the resulting vector in $\F_2^n$ as before.

	\item $\mathrm{Dec}_{\mathrm{CSGN}}$: To decrypt a ciphertext $c$ using the secret key $sk$, output $0$ if $c$ has at least one component equal to $0$ corresponding to an index from $S$ and $1$, otherwise.
\end{itemize}

\noindent The output of the decryption algorithm can be written as
$$\mathrm{Dec}_{\rm CSGN}(sk,c) = \prod_{i \in S} c_i.$$
Notice that, the decryption map is a homomorphism of monoids from $(\F_2^n, \cdot)$ to the monoid $(\F_2, \cdot)$ with the usual multiplication.

In what follows, we present four variants of bridges from the CSGN scheme, denoted by $\sS$, to various FHE schemes. The latter are going to be denoted by $\sH$. Also, the pairing $\langle \cdot , \cdot \rangle : R^n \times R^n \rightarrow R$ will always be the standard inner product over the ring $R$.

\subsection{$1^{\text{st}}$ bridge}
Let $\sH$ be any FHE scheme with plaintext space the field with two elements; hence, the map $\iota$ is the identity map.  
The secret key $sk_{\sS}$ can be represented by the $n$-dimensional standard vectors $e_i$, where $i \in S$. The bridge key generation algorithm encrypts each entry of the vectors $e_i$, $i \in S$ using $pk_{\sH}$ to obtain the bridge key $ bk = \{\widetilde{e_1},..., \widetilde{e_s}\}$, a set of vectors consisting of the aforementioned encryptions.

We remark that the decryption algorithm of $\sS$ may be written as
\begin{equation*}
\mathrm{Dec}_{\sS}(sk_{\sS},c) = \prod_{i \in S} \langle c,e_i \rangle,
\end{equation*}

\noindent so that the bridge algorithm $f$ is as follows:

\begin{equation*}
f(bk,c) = \prod_{i=1}^{s} \langle c,\widetilde{e_i} \rangle = \prod_{i=1}^{s} \left( \sum_{c[j] =1} \tilde{e_i}[j] \right).
\end{equation*}

\noindent For simplicity, we chose the trivial encryptions as the encryptions of the bits of $c$ with $\sH$.

\subsection{$2^{\text{nd}}$ bridge}
We are in the same setting as before, where both plaintext spaces are $\F_2$. Recall that the secret key $sk_{\sS}$ is the characteristic function of the set $S$, represented as an $n$-dimensional bit vector. Then, the decryption of $\sS$ can be alternatively written as
\begin{equation*}
\mathrm{Dec}_{\sS}(sk_{\sS},c) = \prod_{i = 1}^n \Big(1 - (1- c[i]) sk_{\sS}[i] \Big) = \prod_{c[i] = 0} (1 - sk_{\sS}[i])  .
\end{equation*} 

The bridge key $bk$ is constructed as $bk:=\{\widetilde{sk_{\sS}}[1], ..., \widetilde{sk_{\sS}}[n]\}$, where for every $i$, $\widetilde{sk_{\sS}}[i]$ is an encryption of $1-sk_{\sS}[i]$ under $pk_{\sH}$.
Finally, the bridge is given by
\begin{equation*} \label{2ndbringe}
f(bk,c) = \prod_{c[i] = 0} \widetilde{sk_{\sS}}[i].
\end{equation*}

\begin{remark}
The last formula shows that this bridge can be constructed even if the scheme $\sH$ is homomorphic only with respect to multiplication. For example, it can be used when $\sH = \sS$ obtaining something that resembles the key-switching technique in some FHE schemes.
\end{remark}

\subsection{$3^{\text{rd}}$ bridge}

Here, the scheme $\sH$ can be any FHE scheme with plaintext space the finite field $\mathbb{F}_p$, where $p$ is a prime (for example the BGV and B/FV schemes, see \cite{BGV12}, \cite{B12} and \cite{FV12}).

The bridge key generation algorithm instantiates $\mathrm{KeyGen}_{\sS}(1^\lambda)$ and then $\mathrm{KeyGen}_{\sH}(1^{\lambda})$, assuring that the characteristic of $\sP_{\sH}$ is larger than the Hamming weight of $sk_{\sS}$, that is $p> s$. It then chooses positive integers $x_1, ..., x_s$ such that $p= 1 + x_1+ \dots + x_s$, and fixes a bijection  $\varphi : S \rightarrow \{1, ..., s\}$. Consider the vector $sk \in \mathbb F_p^n$, where $sk[i]=0$ if $sk_{\sS}[i]=0$ and $sk[i]=x_{\varphi(i)}$, otherwise. For every $i \in \{1,\dots,n \}$, write $\widetilde{sk}[i]$ for an encryption of $sk[i]$ under $pk_{\sH}$. In this case, the bridge key $bk$ is the set of $\sH$ encryptions $bk=\{\widetilde{sk}[1], \dots, \widetilde{sk}[n] \}.$

We remark that if $\iota: \mathbb{F}_2 \hookrightarrow \mathbb{F}_p$ denotes the usual embedding, then the decryption of $\sS$ satisfies
\begin{equation*}
\mathrm{Dec}_{\sS}(sk_{\sS},c) = \iota^{-1} \Big( 1 - \big(1 + \<c, sk \rangle_{\mathbb{F}_p}\big)^{p-1}\Big).
\end{equation*} 

\noindent The bridge map is defined as
\begin{equation*}
f(bk,c) = \mathrm{Enc}_{\sH}(pk_{\sH},1) - \left( \mathrm{Enc}_{\sH}(pk_{\sH},1) + \sum_{c[i]=1} \widetilde{sk}[i] \right)^{p-1},
\end{equation*}
where the additions, subtractions and exponentiation on the right hand side are homomorphic operations on the ciphertexts of $\sH$.

\begin{remark}
As mentioned in the discussion following Definition \ref{bridge}, one can develop a theory of bridges for which the plaintext spaces of the two encryption schemes vary with $\lambda$ along the same lines. The bridge constructed here falls in this category because the plaintext space of $\sH$ is chosen after the size of the secret key is selected, as part of the Setup/KeyGen algorithm. 
\end{remark}

\subsection{$4^{\text{th}}$ bridge}
This bridge is based on an idea used  in \cite{ASP14} for the bootstrapping procedure of the GSW scheme. 
 Notice that if $c$ is a ciphertext in $\sS$, encrypted using $pk_{\sS}$, then $c$ decrypts to $1$ if and only if the inner product $\langle c, sk_{\sS} \rangle_{\mathbb Z} = s$, namely
\begin{equation*}
\mathrm{Dec}_{\sS}(sk_{\sS},c) = [\<c, sk_{\sS} \rangle_{\mathbb Z}=s],
\end{equation*}

\noindent where $[x=y]$ is, as before, the equality test.

We observe that in the computation of the inner product $\<c, sk_{\sS}\rangle_{\mathbb Z}$ one uses only the additive structure of $\mathbb{Z}$ (also $\mathbb{Z}_m$ with $m > s$ would be sufficient for our purposes).
To find a representation of the cyclic group  $(\mathbb{Z}_m, +)$, one needs first to embed it into the symmetric group $\mathfrak{S}_m$. 
The generator $1 \in \mathbb{Z}_m$ is sent by this injective homomorphism to the cyclic permutation $\pi_1 \in \mathfrak{S}_m$, defined as $\pi_1(i)=i+1$ for $1 \leq i < m$ and $\pi_1(m)=1$. On the other hand, the group $\mathfrak{S}_m$ is isomorphic to the multiplicative group of $m$-by-$m$ permutation matrices, that is matrices with 0 or 1 entries, having exactly one nonzero element in each row and each column. The isomorphism maps the permutation $\pi \in \mathfrak{S}_m$ to the matrix $M_{\pi} = [e_{\pi(1)}, ..., e_{\pi(m)}]$, where $e_i \in \{0, 1\}^m$ is the $i^{\text{th}}$ standard basis vector. The composition of these two homomorphisms gives us an embedding for the cyclic group  $(\mathbb{Z}_m, +)$. For implementation purposes, it is good to notice that the permutation matrices in the image of this embedding can be represented more compactly by just their first column, because the remaining columns are just the successive cyclic shifts of this column. 

Let us explain how the bridge is constructed. Let $m=s+1$ and take $sk = (sk[1], ..., sk[n])$ to be the aforementioned representation of the secret key $sk_{\sS}$, that is $sk[i] = M_{\pi_1}$ if $sk_{\sS}=1$ and $sk[i]$ is the identity matrix otherwise. 
Set $\widetilde{sk}[i]$ to be an encryption of $sk[i]$ under $pk_{\sH}$ for all $i \in \overline{1, n}$, meaning that we encrypt with $\sH$ each entry of the matrix $sk[i]$.
The bridge key $bk$ consists of $\{ \widetilde{sk}[1], \dots, \widetilde{sk}[n]\}$.

The algorithm $f$ takes as input $bk$ and $c$ and computes the matrix
\begin{equation*}
 P^c := \prod_{c[i] = 1} \widetilde{sk}[i],
\end{equation*}
where the right hand side is a product of encrypted matrices, performed homomorphically in $\sC_{\sH}$. We remark that the last entry of the first row of $P^c$ is an encryption of the value returned by the equality test $[\<c, sk_{\sS} \rangle_{\mathbb Z}=s]$. Consequently, we let the output of the bridge map be
\begin{equation*}
f(bk,c) := P^c_{1,s+1}.
\end{equation*}

\section{Entangled encryption schemes}

%TODO

Informally, we say that two encryption schemes $\sS$ and $\sH$ are {\it entangled} if there is a bridge with empty bridge key from one to another.

In this appendix we give an example of such  a bridge. In  this example, the secret key of $\sS$ and $\sH$ are identical.

We believe that whenever two encryption schemes $\sS$ and $\sH$ are entangled, there is a relation between the ensembles of distributions of their secret keys. We regard this as an interesting question for future research.

The presented bridge does not follow Gentry's recipe. We start by recalling the Goldwasser-Micali and Sander-Young-Yung encryption schemes. A bridge from the former to the latter is then presented. The presentation is followed by an interesting application of this bridge.

\subsection{Goldwasser-Micali Cryptosystem}

The Goldwasser-Micali encryption scheme  is an asymmetric key encryption algorithm developed by Shafi Goldwasser and Silvio Micali in \cite{GM84}. If $p,q$ are two primes and $N = p \cdot q$, then let
$J_1(N):=\{ x \in (\mathbb Z/ N \mathbb Z)^\times | \left( \frac{x}{N} \right) = 1 \}$ be the multiplicative group of invertible integers modulo $N$ with Jacobi symbol equal to $1$.
The GM-encryption scheme $(\mathbb{Z}/2\mathbb{Z}, J_1(N), {\rm KeyGen}_{GM}, {\rm Enc}_{GM}, {\rm Dec}_{GM})$ is given as follows:

\begin{itemize}
\item {\bf KeyGen}($1^{\lambda}$): Choose two primes $p=p(\lambda), q=q(\lambda)$ of size $\lambda$ and let $N=pq$. Choose $\eta \in  (\mathbb{Z}/N\mathbb{Z})^\times$ such that $\left( \frac{\eta}{p} \right) = \left( \frac{\eta}{q} \right) = -1$, which yields that $\eta \in J_1(N)$.
The public key is the pair $(N, \gamma:=\eta\cdot u^ 2),$ where $u$ is a random element of $(\mathbb{Z}/ N \mathbb{Z})^\times$.
The secret key is the pair $(p,q)$.
\item {\bf Enc}: To encrypt $m \in \mathbb{Z}/2 \mathbb{Z}$, choose a random $\xi \in \mathbb{Z}/ N \mathbb{Z}$ and let ${\rm Enc}_{GM}(m)= \gamma^m \xi^2$.
\item {\bf Dec}: To decrypt $c \in J_1(N)$, compute the Jacobi symbol $\left( \frac{c}{p} \right)$. Set ${\rm Dec}_{GM}(c)=0$ if the answer is $1$ and ${\rm Dec}_{GM}(c) = 1$ if the answer is $-1$.
\end{itemize}

\noindent The GM-encryption scheme is homomorphic with respect to addition in $\mathbb{Z}/2\mathbb{Z}$ and multiplication in $J_1(N)$, i.e. 
\begin{equation*}
{\rm Dec}_{GM}(c_1 \cdot c_2) = {\rm Dec}_{GM}(c_1) + {\rm Dec}_{GM}(c_2)
\end{equation*}
\noindent for all $c_1, c_2 \in J_1(N)$.

\subsection{The Sander-Young-Yung Cryptosystem} 

In this part of the appendix we present a homomorphic encryption scheme over the multiplicative monoid $(\mathbb{Z} / 2 \mathbb{Z}, \cdot)$ introduced in \cite{SYY}. To describe the scheme we shall use the encryption scheme of Goldwasser-Micali, which was recalled above.

\begin{itemize}
\item {\bf Keygen}($1^{\lambda}$): Choose two primes $p=p(\lambda)$, $q=q(\lambda)$ as in the Goldwasser-Micali scheme. Choose $\ell = \ell(\lambda)$ of size $\Theta(\lambda)$. Compute $N=pq$. The public key and secret keys are the same as in the Goldwasser-Micali scheme. 
\item {\bf Enc}: If $m=1$ set $v = (0, ..., 0) \in \{0, 1\}^{\ell}$. If $m=0$ set $v = (v_1,...,v_n) \in  \{0, 1\}^{\ell}$, where the
components $v_i$ are randomly chosen in $\{0, 1\}$, not all equal to $0$. Encrypt each component of $v$ with the Goldwasser-Micali scheme to get a vector in $\mathscr{C}_{SYY}:= J_1(N)^{\ell}$.
\item {\bf Dec}: To recover the plaintext from the ciphertext $c \in \mathscr{C}$, first decrypt each component of $c$ using the decryption algorithm of the Goldwasser-Micali scheme, and then if the obtained vector is the $0$-vector the message decrypts to 1, else to 0.
\end{itemize}

\noindent Let us describe an operation $\odot$ on the ciphertext space $\mathscr{C}_{SYY}$. If $x$ and $y$ are two ciphertexts then $z:= x \odot y$ is defined as follows:

\medskip

\noindent 1. Choose uniformly at random two $\ell \times \ell$ matrices over $\mathbb{Z}/2 \mathbb{Z}$ until two nonsingular
matrices $A=(a_{ij})$ and $B=(b_{ij})$ are found.

\medskip

\noindent 2. If $x=(x_1, ..., x_{\ell})$, $y=(y_1, ..., y_{\ell})$, then compute

$$z_i = \displaystyle{\prod_{j, a_{ij}=1} x_j \cdot \prod_{j, b_{ij}=1} y_j}$$

\noindent for all $i$. 

\medskip

\noindent 3. Pick uniformly at random $r_1, ..., r_{\ell} \in (\mathbb{Z}/ N \mathbb{Z})^\times$ and
set $z= (z_1 r_1^2,..., z_{\ell} r_{\ell}^2)$.

\medskip

\noindent Let us denote by $v_c$ the bit vector obtained by applying the decryption algorithm of the Goldwasser-Micali scheme componentwise to the ciphertext $c \in\mathscr{C}$. If $z:= x \odot y$ then Step 2 above is equivalent to:

\begin{equation*}
v_z= A v_x + B v_y,
\end{equation*}

\noindent where the operations are the usual addition and multiplication in $\mathbb{Z}/ 2\mathbb{Z}$. Notice that ${\rm Dec}_{SYY}(z) \neq {\rm Dec}_{SYY}(x) \cdot {\rm Dec}_{SYY}(y)$ if and only if $A v_x + B v_y = \vec 0$ (here $\vec 0$ is the zero vector in $(\mathbb{Z} / 2\mathbb{Z})^{\ell}$), and $v_x \neq \vec 0$, $v_y \neq \vec 0$.
Since $v_x \neq \vec 0$ and $A$ is nonsingular, the product $A v_x$ can be any nonzero vector in $(\mathbb{Z} / 2\mathbb{Z})^{\ell}$, and in fact any such vector occurs with the same probability. Of course, the same is true for $B v_y$ such that the situation described above occurs with probability $\leq \dfrac{1}{2^{\ell}}$. In other words, except with exponentially small probability, we have that

\begin{equation*}
 {\rm Dec}_{SYY} (x \odot y) = {\rm Dec}_{SYY}(x) \cdot {\rm Dec}_{SYY}(y).
\end{equation*}

\subsection{A bridge from GM to SYY}

Here, we construct a bridge from the Goldwasser-Micali encryption scheme to the Sander-Young-Yung encryption scheme. 
After generating a secret key $(p, q)$ of GM, the key generation algorithm of the bridge sets the same pair $(p,q)$ as the secret key for the SYY encryption scheme. Then, the public keys for the two encryption schemes are generated independently using their respective key generation algorithms. After that, the bridge key generation algorithm does not output anything, i.e. the support of the distribution $BK$ is the empty set. 

Now, for $c \in J_1(N)$, choose uniformly at random a non-singular matrix $A \in {\rm GL}_{\ell} (\mathbb{Z} / 2\mathbb{Z})$ and
compute

\begin{equation*}
t_i =  \prod_{j, a_{ij}=1} c \gamma' = \left( c \gamma' \right)^{|\{j | a_{ij}=1 \}|}
\end{equation*}
for all $i \in \overline{1, \ell}$, where $\gamma'$ is the second component of the public key of the SYY scheme. Pick uniformly at random $r_1, \dots, r_{\ell} \in (\mathbb{Z}/ N \mathbb{Z})^{\times}$ and
set

\begin{equation*}
f(c)  =  (t_1 r_1^2, ..., t_{\ell} r_{\ell}^2).
\end{equation*}
 
\noindent If ${\rm Dec}_{GM}(c) = 1$, then ${\rm Dec}_{GM}(c \gamma') = 0$ so that ${\rm Dec}_{GM}(t_i) = 0$, $\forall i$. Therefore, $v_{f(c)} = \vec 0$ and hence ${\rm Dec}_{SYY} (f(c)) = 1$. On the other hand, if ${\rm Dec}_{GM}(c) = 0$, then ${\rm Dec}_{GM}(c \gamma') = 1$, and since $A$ is nonsingular there exist $i \in \overline{1, \ell}$ such that ${\rm Dec}_{GM}(t_i) = 1$. We get that $v_{f(c)} \neq \vec{0}$, equivalently 
${\rm Dec}_{SYY} (f(c)) = 0$.

\begin{remark} \label{securityGM}
The security of this bridge reduces to the security of the GM scheme (see \cite{GM84}) using Theorem \ref{mt}. Indeed, the bridge key distribution is empty, thus trivially polynomial-time constructible on fibers. On the other hand, the security of SYY encryption scheme can be easily reduced to the security of GM (see \cite{SYY}). Alternatively, one can use Theorem \ref{knowledge:thm} instead of \ref{mt}. To see this, note that in the notation of Section 3, the public key of the scheme attached to this bridge $PK_{\sG_f}$ consists of just GM's public key and the security of $GM[PK_{\sG_f}]$ is equivalent to the security of GM.
\end{remark}

\subsection{An application}

As an application of the above bridge we show that the comparison circuit can be evaluated homomorphically. For this, let $\vec{x}=(x_1, x_2, ..., x_n)$ and $\vec{y}=(y_1, y_2, ..., y_n)$ be two bit vectors. The two vectors coincide if and only if

\begin{equation*}
(x_1 + y_1 + 1) \cdot ... \cdot (x_n + y_n + 1) = 1,
\end{equation*} 

\noindent so that the comparison circuit $[\vec{x}=\vec{y}]$ is defined by 
$$
[\vec{x}=\vec{y}]:=(x_1 + y_1 + 1) \cdot ... \cdot (x_n + y_n + 1).
$$

\noindent Suppose now that $\vec{c} = (c_1, ..., c_n)$ and $\vec{d} = (d_1, ..., d_n)$ are encryptions of the vectors $\vec{x}$, $\vec{y}$ with the Goldwasser-Micali cryptosystem. To homomorphically evaluate the comparison circuit, we compute:
\begin{equation*}
{\rm Eval}([\vec{x}=\vec{y}], \vec{c},\vec{d} \;):=  \bigg( \Big( \big( f(c_1 \cdot d_1 \cdot \gamma) \odot f(c_2 \cdot d_2 \cdot \gamma) \big) \odot  ... \Big) \odot f (c_n \cdot d_n \cdot \gamma) \bigg).
\end{equation*}
\noindent Notice that ${\rm Dec}_{SYY}\left( {\rm Eval}([\vec{x}=\vec{y}], \vec{c}, \vec{d} \;) \right) = [\vec{x}=\vec{y}]$, except with negligible probability in the security parameter.

\medskip

We end this appendix with the following reflection. When two encryption schemes admit the construction of a bridge which has an empty bridge key, this may be interpreted as some sort of entanglement between the schemes. Along the same line of thought, if one can prove that such a bridge cannot be constructed, the encryption schemes may be regarded as being independent.

\vskip-6pt
\section{Experiments}
\vskip-6pt
We conducted experiments for the bridges described in Appendices A and B. For each of the four different bridges in Appendix A, we compare the results of the homomorphic evaluation of a circuit consisting of only one monomial in the following two ways. First, we encrypt each factor of the monomial and perform the homomorphic multiplications of these factors using the $\mathrm{CSGN}$ scheme. Then, bridges described in Appendix A are applied, in turn, to obtain a ciphertext in a fully (leveled) homomorphic encryption scheme based on (R)LWE. We compare this to the alternative option of evaluating the monomial directly on encryptions in the FHE scheme. If the degree of the monomial is larger than a certain threshold, the first procedure outperforms the second in terms of speed. We identified this threshold for each of the FHE schemes in which we performed experiments.

These computations were carried on a virtual machine having an Intel CPU (I7-4770, 4 cores, 3.4 GHz, 12 GB RAM), using a single threaded implementation. Table 1 consists of an overview of the processing times for each bridge using the implementations of BGV, BFV and TFHE schemes, namely the HElib \cite{HELIB}, SEAL \cite{SEAL} and TFHE \cite{CGGI19} software libraries. In the first two columns of the table, one can find the version of the bridge that was implemented, the FHE target scheme and the security parameters for the two schemes. The timings are measured such that all encryptions maintain approximately the same security level $\lambda$ and listed in the last two columns. The small variation in $\lambda$ is due to parameter tuning in the different software libraries.

\begin{table}[ht]
\caption{Bridge evaluation \label{tab:comp}}
\centering{}
\begin{tabular}{|c|c|c|c|}
\hline
{\textbf{Bridge (CSGN-$\lambda$) }}  & $\textbf{LWE ($\lambda$) }$ & $\textbf{ENC (Bridge key) }$ & $\textbf{Bridge time}$ \tabularnewline
\hline
{$1^{\text{st}}$(125)}  & $\mathrm{BGV}(121)$ & $69 \ sec$  & $2.6 \ sec$   \tabularnewline
\hline
{$1^{\text{st}}$(125)}  & $\mathrm{TFHE}(128)$ & $186 \ ms$  & $38.33 \ sec$ \tabularnewline
\hline

{$1^{\text{st}}$(125)}  & $\mathrm{BFV}(128)$ & $38.97 \ sec$  & $209.95 \ ms$ \tabularnewline
\hline

{$2^{\text{nd}}$(125)}  & $\mathrm{BGV}(114)$ & $14.6 \ sec$  & $68.28 \ sec$  \tabularnewline
\hline
{$2^{\text{nd}}$(125)}  & $\mathrm{TFHE}(128)$ & $2.94 \ ms$  & $1049 \ ms$  \tabularnewline
\hline

{$2^{\text{nd}}$(125)}  & $\mathrm{BFV}(128)$ & $698 \ ms$  & $2.24 \ sec$ \tabularnewline
\hline

{$3^{\text{rd}}$(120)}  & $\mathrm{BGV}(145)$ & $7.65  \ sec$  & $248 \ ms$  \tabularnewline
\hline

{$3^{\text{rd}}$(120)}  & $\mathrm{BFV}(128)$ & $8.2 \ sec$  & $156.46 \ ms$ \tabularnewline
\hline

{$4^{\text{th}}$(115)}  & $\mathrm{TFHE}(128)$ & $162.6\ ms$  & $989.4 \ sec$  \tabularnewline
\hline
\end{tabular}
\end{table}

The reason we are missing an implementation for our third bridge using the TFHE library comes from the lack of flexibility in choosing as plaintext space a ring of characteristic $p>2$ in this library. Additionally, we felt that adapting the TFHE library was beyond the scope of our work.  Also, the timing for running the fourth bridge in BGV and BFV could not be measured because of large memory usage, which exceeded the virtual machine RAM. Moreover, regarding the fourth bridge, the implementation is optimized to store only the first column of each associated bit in the secret key, while the matrix multiplications involve only homomorphic algebraic operations on encryptions from the first column of the matrices.

There is no doubt that homomorphically evaluating a circuit whose polynomial representation has a large number of monomials of low degree using the bridge is inefficient and there is little hope for optimizations in terms of speed. However, if some monomials have large degree, one might choose to do so, because first performing multiplications in the CSGN scheme, followed by additions in the (R)LWE setting might result in lower noise growth. Moreover, by increasing the multiplicative depth of the circuit, we observe that its evaluation is faster using the bridge than evaluating the circuit entirely in the (R)LWE schemes. This can be observed in the figures below.

Since the multiplication in the CSGN scheme is inexpensive, the evaluation time in the bridge using BGV, BFV  and TFHE is almost constant as it essentially consists only of the evaluation time of the bridge algorithm for one CSGN ciphertext. Small variations in execution time for the bridge are due to the CPU scheduling process. The drops in evaluation times occur when the instruction-specific and data-specific cache at different levels in the CPU is filled with numerous repetitive instructions. The timings for evaluating the circuit entirely in the BGV or BFV scheme grow linearly with the degree of the monomial. We notice that in the TFHE case, the running time of the evaluation starts growing exponentially in the number of multiplications, at some point. This is explained by the fact that the TFHE software library goes automatically into bootstrapping, whereas in the HElib and SEAL  software libraries we can choose parameters in which one can evaluate the circuit without the costly bootstrapping procedure.
\begin{center}
\begin{figure}[ht]
\begin{center}
	\includegraphics[scale=0.30]{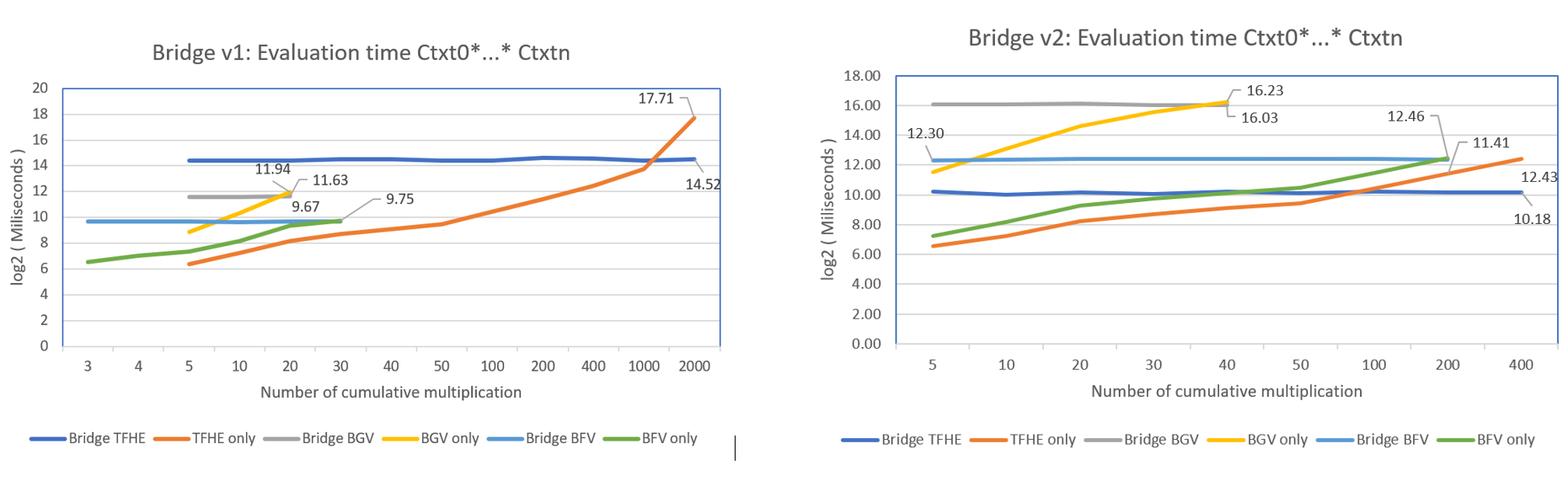}
\end{center}
	\caption{The first and second bridges}
\end{figure}
\end{center}

\begin{center}
	\begin{figure}[ht]
	\label{b3_bgv_bfv}
	\begin{center}
	\includegraphics[scale=0.45]{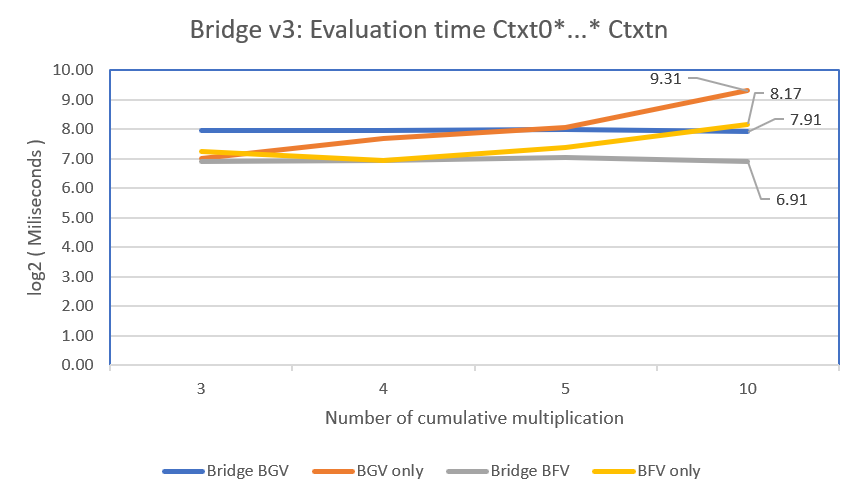}
	\end{center}
	\caption{The third bridge - BGV \& BFV}
	\end{figure}
\end{center}
We now report on the implementation of the bridge from the Goldwasser-Micali encryption scheme to the Sander-Young-Yung encryption scheme constructed in the Appendix B. In the table below, one can find the timings required for running the bridge, as well as the ones needed for the homomorphic evaluation of the comparison circuit. The measurements were performed on an Intel I7-1068NG7 CPU laptop with 32GB of RAM. Since the parameter $\ell$ of the SYY scheme does not have an impact on the security, but rather on the probability to correctly decrypt the ciphertext  $\left(\geq 1- \dfrac{1}{2^{\ell}}\right)$, we fix $\ell$ to be 50.

\begin{table}
\centering{}
\caption{Homomorphic evaluation of comparison circuit using GM-SYY bridge \label{tab:comp2}}
\begin{tabular}{|c|c|c|c|c|c|}
\hline
 \;\;\; $n$ \;\;\;  & \;\; $ \log_2(N)$ \;\; & \; \; \; $\textbf{GM $\cdot$}$ \; \; \; & \; \; $\textbf{SYY  $\odot$}$  \; \; & $\; \; \textbf{GM $\rightarrow$ SYY}$ \; \; &\textbf{$[\vec{x} = \vec{y}]$} \tabularnewline
\hline
 4 & 1024 & $0.002 \ ms$  & $10.02 \ ms$ & $4.54 \ ms$ & $58.35 \ ms$   \tabularnewline
\hline
 4 & 2048 & $0.003 \ ms$  & $29.96 \ ms$ & $11.64 \ ms$ & $164.44 \ ms$   \tabularnewline
\hline
 4 & 4096 & $0.008 \ ms$  & $84.01 \ ms$ & $32.82 \ ms$ & $467.43 \ ms$   \tabularnewline
\hline
 8 & 1024 & $0.003 \ ms$  & $10.70 \ ms$ & $4.77 \ ms$ & $123.98 \ ms$   \tabularnewline
\hline
 8 & 2048 & $0.004 \ ms$  & $30.12 \ ms$ & $11.89 \ ms$ & $336.27 \ ms$   \tabularnewline
\hline
 8 & 4096 & $0.008 \ ms$  & $84.65 \ ms$ & $33.46 \ ms$ & $945.26 \ ms$   \tabularnewline
\hline
 16 & 1024 & $0.002 \ ms$  & $10.8 \ ms$ & $4.87 \ ms$ & $251.44 \ ms$   \tabularnewline
\hline
 16 & 2048 & $0.004 \ ms$  & $29.69 \ ms$ & $11.55 \ ms$ & $660.17 \ ms$   \tabularnewline
\hline
16 & 4096 & $0.008 \ ms$  & $85.49 \ ms$ & $33.71 \ ms$ & $1907.78 \ ms$   \tabularnewline
\hline
32 & 1024 & $0.003 \ ms$  & $10.4 \ ms$ & $4.69 \ ms$ & $484.10 \ ms$   \tabularnewline
\hline
32 & 2048 & $0.004 \ ms$  & $30.29 \ ms$ & $11.82 \ ms$ & $1348.44 \ ms$   \tabularnewline
\hline
32 & 4096 & $0.009 \ ms$  & $82.51 \ ms$ & $32.34 \ ms$ & $3576.41 \ ms$   \tabularnewline
\hline
\end{tabular}
\medskip
\end{table}
The parameters $n$ and $N$ in Table \ref{tab:comp2} stand for the bit-lengths of $\vec{x}, \vec{y}$ and, respectively, the Goldwaser-Micalli modulus. The timings required for the one homomorphic operation in each scheme can be found in the third and the fourth columns. We notice that the timings presented above grow linearly with the number of bits required to represent the input data. This can be observed in the following figure. 
\begin{center}
	\begin{figure}[ht]
	\label{gm_syy_timings}
	\begin{center}
	\includegraphics[scale=0.3]{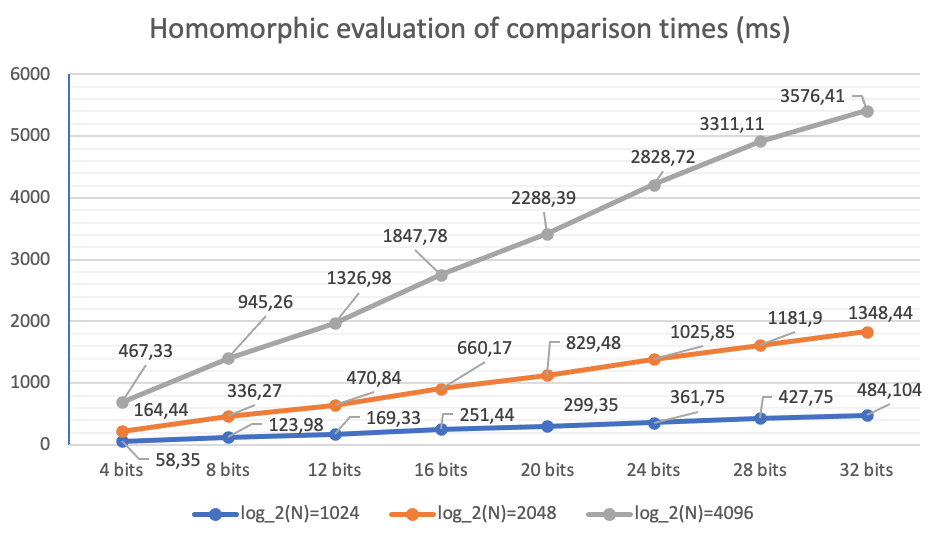}
	\end{center}
	\caption{Evaluation times for the comparison circuit using GM-SYY bridge}
	\end{figure}
\end{center}

%% The Appendices part is started with the command \appendix;
%% appendix sections are then done as normal sections
%% \appendix

%% \section{}
%% \label{}

%% References
%%
%% Following citation commands can be used in the body text:
%% Usage of \cite is as follows:
%%   \cite{key}         ==>>  [#]
%%   \cite[chap. 2]{key} ==>> [#, chap. 2]
%%

%% Authors are advised to use a BibTeX database file for their reference list.
%% The provided style file elsarticle-num.bst formats references in the required Procedia style

%% For references without a BibTeX database:

% \begin{thebibliography}{00}

%% \bibitem must have the following form:
%%   \bibitem{key}...
%%

% \bibitem{}

% \end{thebibliography}

\end{document}